\documentclass{amsproc}
\usepackage{amssymb}
\usepackage{graphicx}

\newtheorem{theorem}{Theorem}[section]
\newtheorem{proposition}[theorem]{Proposition}
\newtheorem{corollary}[theorem]{Corollary}
\newtheorem{lemma}[theorem]{Lemma}
\theoremstyle{definition}

\newtheorem{example}[theorem]{Example}

\theoremstyle{remark}
\newtheorem{remark}[theorem]{Remark}
\numberwithin{equation}{section}

\newcommand{\cod}{{\mathcal C}}
\newcommand{\calp}{{\mathcal P}}
\newcommand{\calx}{{\mathcal X}}
\newcommand{\cald}{{\mathcal D}}
\newcommand{\call}{{\mathcal L}}
\newcommand{\calh}{{\mathcal H}}
\newcommand{\cals}{{\mathcal S}}
\newcommand{\fq}{\mathbb{F}_q}
\newcommand{\fqd}{\mathbb{F}_{q^2}}
\newcommand{\nn}{\mathbb{N}}
\newcommand{\pp}{\mathbb{P}}
\newcommand{\bx}{{\mathbf x}}
\newcommand{\bb}{{\mathbf b}}

 \newcommand{\bc}{{\mathbf c}}
 \newcommand{\be}{{\mathbf e}}
\newcommand{\bh}{{\mathbf h}}

 \newcommand{\bs}{{\mathbf s}}
 \newcommand{\bu}{{\mathbf u}}
 \newcommand{\bv}{{\mathbf v}}
 
 \newcommand{\bz}{{\mathbf z}}
\newcommand{\bcero}{{\mathbf 0}}
 \newcommand{\bH}{{\mathbf H}}
 \newcommand{\bB}{{\mathbf B}}
 
 \newcommand{\bD}{{\mathbf D}}
\newcommand{\bS}{{\mathbf S}}

 \newcommand{\sop}{{\rm supp}}
\newcommand{\divi}{{\rm div}}
\newcommand{\gaps}{{\rm Gaps}}
 
\newcommand{\rank}{{\rm rank}}

\begin{document}

\title[An introduction to Algebraic Geometry codes]{An introduction to Algebraic Geometry codes}
\author{Carlos Munuera}
\address{University of Valladolid}
\curraddr{Avda Salamanca SN, 47014 Valladolid, Castilla, Spain}
\email{cmunuera@arq.uva.es}
\author{Wilson Olaya-Le\'on}
\address{Universidad Industrial de Santander}
\curraddr{Cra 27, Cll 9, AA 678 Bucaramanga, Santander, Colombia.}
\email{wolaya@uis.edu.co}

\begin{abstract}
We present an introduction to the theory of algebraic geometry codes.
Starting from evaluation codes and codes from order and weight functions, special attention is given to one-point codes and, in particular, to the family of Castle codes. 
\end{abstract}

\maketitle

\section{Introduction}

Let $\fq$  be a finite field with $q$ elements. A {\em linear code of length $n$ and dimension $k$} over $\fq$, a $[n,k]$ code for short, is a $k$-dimensional linear space $\cod\subseteq\fq^n$. The {\em minimum distance} of $\cod$ is by definition
\begin{equation*}
d=\min\{ d(\bu,\bv): \bu,\bv\in\cod, \bu\neq\bv\}=\min\{wt(\bu):\bu\in\cod,\bu\neq \bcero\}
\end{equation*}
where $d$ stands for the Hamming distance, $d(\bu,\bv)=\#\{i : u_i\neq v_i\}$, and $wt$ for the Hamming weight, $wt(\bu)=d(\bu,\bcero)$. A \lq\lq good code" is one that optimizes simultaneously the ratios $d/n$ and $k/n$. 

The problem of finding good codes is central to the theory of error correcting codes. For many years coding theorists have addressed this problem by adding more and more algebraic and combinatorial structure to $\cod$. In particular, codes with excellent properties have been obtained by using techniques and resources from algebra and algebraic geometry, the so-called {\em algebraic geometry} codes. Most of these techniques are highly specialized and the study of the obtained codes is very elegant  but in general difficult. Indeed, given such a code, often it is not possible to calculate its exact minimum distance, and sometimes even its dimension. 

In this chapter we present a short introduction to  algebraic geometry codes. We use the order bounds on the minimum distance as a motivation to introduce evaluation and algebraic geometry codes. Then we center our attention on one-point codes, and later on the family of Castle codes. As a result of this orientation we can overview quickly much of the basic theory. However we warn the reader that many important parts and facts have been omitted. For a complete treatment we refer to the excellent texts \cite{HLP} and \cite{StichtenothLibro}. The canonical reference for general error correcting codes is the very complete book \cite{MacSlo} (although it does not contain the theory of AG codes).

\section{The order bounds on the minimum distance}\label{Sect1}

\subsection{Bounds}

As noted above,  computing the true minimum distance $d$ of a linear code $\cod$ is in general a difficult problem (it is an NP-complete problem, see \cite{npcompleto}).  Often we have to settle for an estimate of $d$ based on some available lower bound. And then evaluate the quality of our parameters by comparing them with several upper bounds.    Usually upper bounds are general, valid for all linear codes.  Let us show an important example.

\begin{theorem}[Singleton bound]\label{Singleton}
The parameters $n,k,d$ of a linear code $\cod$ verify $k+d\le n+1$.
\end{theorem}
\begin{proof}
Let $\pi:\cod\rightarrow\fq^{n-d+1}$ be the projection obtained by deleting $d-1$ fixed coordinates. Since each codeword of $\cod$ has at least $d$ nonzero coordinates, $\pi$ is an injective linear map, hence $\dim(\pi(\cod))=k$ and thus $k\le n-d+1$.  
\end{proof}

Codes reaching equality in the Singleton bound are called {\em maximum distance separable} (or MDS) codes.

Lower bounds on the minimum distance are designed to be applied to some particular families or constructions of codes. Significant examples could be BCH and Goppa bounds (BCH and algebraic geometry codes  respectively).
Besides uniform ones, other interesting lower bounds are of
order type. They are based on obtaining different estimates for different
subsets of codewords. Such a bound is successful if for each subset we can find estimates
better than a uniform bound for all codewords. In this chapter we shall explain two bounds of this type.

\subsection{$\fq$-algebras}

Throughout this chapter, an $\fq$-{\em algebra} will be a commutative ring $R$ with a unit, containing $\fq$ as a subring. Then $R$ is a vector space over $\fq$. The most interesting examples of $\fq$-algebras are the polynomial ring in $m$ variables $\fq[X_1,\dots,X_m]$ and its quotients $\fq[X_1,\dots,X_m]/I$, where $I$ is an ideal. Other important example is $\fq^n$. Since $\fq$ is naturally isomorphic to  $\{(\lambda,\dots,\lambda)  | \lambda \in \fq\}$, it turns out that $\fq^n$ is also an algebra with the coordinate wise product $*$,  
\begin{equation*}
(u_1,\dots,u_n)*(v_1,\dots,v_n)=(u_1v_1,\dots,u_nv_n).
\end{equation*}
Note that $(\lambda,\dots,\lambda)*(u_1,\dots,u_n)=\lambda (u_1,\dots,u_n)$ hence the ring and vector space structures on $\fq^n$ are fully compatible.

\subsection{The Andersen-Geil bound}\label{cotaAndersen-Geil}

Let ${\mathcal B} =\{ {\mathbf b}_1, \dots, {\mathbf b}_n\}$ be a basis of ${\mathbb F}_q^n$. 
We consider the linear codes $\cod_0=(\bcero )$, and for $k=1,\dots,n$,  
\begin{equation*}
\cod_k=\langle {\mathbf b}_1,\dots,{\mathbf b}_k \rangle .
\end{equation*} 
$\cod_k$ is a $[n,k]$ code.
Associated to the chain $\cod_0=(\bcero)\subset\cod_1\subset\cdots\subset\cod_n=\fq^n$, we define the sorting map $\rho_{\mathcal B}:{\mathbb F}_q^n \rightarrow \{0,\dots,n\}$
by $\rho_{\mathcal B}({\mathbf v})=\min \{ r : {\mathbf v} \in \cod_r\}$. 

\begin{lemma}\label{properties} 
Let ${\mathbf v}_1,\dots,{\mathbf v}_m \in {\mathbb F}_q^n$. Then 
\begin{enumerate}
\item $\rho_{\mathcal B}({\mathbf v}_1+\dots+{\mathbf v}_m) \leq \max \{\rho_{\mathcal B}({\mathbf v}_1), \dots, \rho_{\mathcal B}({\mathbf v}_m)\}$. If there exists $j$ such that $\rho_{\mathcal B}({\mathbf v}_i) <\rho_{\mathcal B}({\mathbf v}_j)$ for all $i \neq j$, then equality holds. 
\item If ${\mathbf v} \neq {\mathbf 0}$ then there exist $\lambda_1,\dots,\lambda_{\rho_{\mathcal B}({\mathbf v})}\in\fq$ with $\lambda_{\rho_{\mathcal B}({\mathbf v})}\neq 0$ such that ${\mathbf v}=\lambda_1{\mathbf b}_1+\dots+\lambda_{\rho_{\mathcal B}({\mathbf v})}{\mathbf b}_{\rho_{\mathcal B}({\mathbf v})}$.
\item $\dim (\langle {\mathbf v}_1,\dots,{\mathbf v}_m \rangle ) \geq \# \{\rho_{\mathcal B}({\mathbf v}_1),\dots,\rho_{\mathcal B}({\mathbf v}_m)\}$. 
Conversely, if $D\subseteq {\mathbb F}_q^n$ is a linear subspace of dimension $m$, then  there exists a basis $\{{\mathbf u}_1,\dots,{\mathbf u}_m\}$ of $D$ such that $\# \{\rho_{\mathcal B}({\mathbf u}_1),\dots,\rho_{\mathcal B}({\mathbf u}_m)\}=m$.
\end{enumerate}
\end{lemma}
 
\begin{proof}
(1) Both statements follow from the linear structure of our codes. (2) follows from (1).(3) Assume $\#\{\rho_{\mathcal B}({\mathbf v}_1),\dots,\rho_{\mathcal B}({\mathbf v}_m)\} = t$ and $\rho_{\mathcal B}({\mathbf v}_1)<\dots <\rho_{\mathcal B}({\mathbf v}_t)$. If $\lambda_1{\mathbf v}_1+\dots+\lambda_t{\mathbf v}_t=0$ then $0=\rho_{\mathcal B}({\mathbf 0})=\rho_{\mathcal B}(\lambda_1{\mathbf v}_1+\dots+\lambda_t{\mathbf v}_t)=\max \{ \rho_{\mathcal B}({\mathbf v}_i) : \lambda_i \neq 0\}$. 
By (1) this implies $\lambda_1= \dots =\lambda_t=0$. Conversely write $D_i=D\cap C_i$. For all $i=1,\dots,n$, it holds that  $D_i=D_{i-1} \oplus (D\cap \langle {\mathbf b}_i \rangle )$, hence $\dim (D_{i-1})\le \dim (D_i)\le \dim (D_{i-1})+1$ and the last inequality is an equality precisely $m$ times. If $D_i\neq D_{i-1}$, take a vector ${\mathbf u}_i\in D_i\setminus D_{i-1}$. Then $\# \{\rho_{\mathcal B}({\mathbf u}_1),\dots,\rho_{\mathcal B}({\mathbf u}_m)\}=m$ and  $\{{\mathbf u}_1,\dots, {\mathbf u}_m\}$ is a basis of $D$. 
\end{proof}

We consider in $\nn^2$ the partial order $(r,s) \prec (i,j)$ if and only if $r \leq i$, $s \leq j$ and $(r,s) \neq (i,j)$.
A pair of nonzero vectors $(\bu,\bv)$  is called {\em well-behaving} (with respect to the basis ${\mathcal B}$) if for any pair $({\mathbf b}_r,{\mathbf b}_s)$ such that $(r,s) \prec(\rho_{\mathcal B}(\bu),\rho_{\mathcal B}(\bv))$ it holds that $\rho_{\mathcal B}({\mathbf b}_r*{\mathbf b}_s)<\rho_{\mathcal B}(\bu*\bv)$. For $i=1,\dots, n$, define the set
\begin{equation*}
\Lambda_i= \{ {\mathbf b}_j \in {\mathcal B} : ({\mathbf b}_i,{\mathbf b}_j) \mbox{ is well-behaving}\}  .
\end{equation*} 
Let  ${\mathbf v} \in {\mathbb F}_q^n$,  ${\mathbf v} \neq {\mathbf 0}$. According to Lemma \ref{properties} (2), we 
can write $\bv$ as a linear combination ${\mathbf v}=\lambda_1{\mathbf b}_1+\dots+\lambda_{\rho_{\mathcal B}({\mathbf v})}{\mathbf b}_{\rho_{\mathcal B}({\mathbf v})}$ with $\lambda_{\rho_{\mathcal B}({\mathbf v})} \neq 0$. Then, if ${\mathbf b}_j \in \Lambda_{\rho_{\mathcal B}({\mathbf v})}$ we have 
\begin{equation*}
\rho_{\mathcal B}({\mathbf v}*{\mathbf b}_j)=\rho_{\mathcal B} (\sum_{i=1}^{\rho_{\mathcal B}({\mathbf v})} \lambda_i{\mathbf b}_i*{\mathbf b}_j)=\rho_{\mathcal B}({\mathbf b}_{\rho_{\mathcal B}({\mathbf v})}*{\mathbf b}_j).
\end{equation*}

\begin{proposition}\label{agbound}
Let  ${\mathbf v} \in {\mathbb F}_q^n$. If ${\mathbf v} \neq {\mathbf 0}$ then $wt({\mathbf v}) \geq \# \Lambda_{\rho_{\mathcal B}({\mathbf v})}$.
\end{proposition}
\begin{proof}  
Consider the space $V({\mathbf v})= \{ {\mathbf u} \in {\mathbb F}_q^n : \mbox{supp} ({\mathbf u}) \subseteq\mbox{supp}({\mathbf v})\} = \{ {\mathbf u}*{\mathbf v} : {\mathbf u} \in {\mathbb F}_q^n\}$.  Then
$wt({\mathbf v})=\dim (V({\mathbf v})) \geq \dim (\langle {\mathbf v}*{\mathbf b}_1,\dots,{\mathbf v}*{\mathbf b}_n \rangle) \geq \#\{\rho_{\mathcal B}({\mathbf v}*{\mathbf b}_1),\dots, \rho_{\mathcal B}({\mathbf v}*{\mathbf b}_n)\} \geq \#\{\rho_{\mathcal B}({\mathbf v}*{\mathbf b}_j) : {\mathbf b}_ j \in \Lambda_{\rho_{\mathcal B}({\mathbf v})}\}=\#\{\rho_{\mathcal B}({\mathbf b}_{\rho_{\mathcal B}({\mathbf v})}*{\mathbf b}_j) : {\mathbf b}_j \in \Lambda_{\rho_{\mathcal B}({\mathbf v})}\}= \#\Lambda _{\rho_{\mathcal B}({\mathbf v})}$.
\end{proof}

This result directly leads to the following bound.

\begin{theorem}\label{AndGeil-primary} 
For $k=1,\dots,n$, the minimum distance of $\cod_k$ satisfies 
\begin{equation*}
d(\cod_k) \geq  \min \{ \#\Lambda_r : r=1,\dots, k\}.  
\end{equation*}
\end{theorem}

The inequality stated in the above theorem is called the {\em Andersen-Geil bound on the minimum distance of the primary code} $\cod_k$, or  {\em order bound with respect to the basis $\mathcal B$ on the minimum distance of the primary code} $\cod_k$.
Note that the sets $\Lambda_r$ depend on the basis $\mathcal B$. So the bound depends on $\mathcal B$ as well.
This bound
can be applied to an arbitrary linear code $\cod$, just by including it into any increasing chain of codes  $\cod_1 \subset\dots\subset \cod_{k-1} \subset \cod \subset \cod_{k+1}\subset\dots \subset \cod_n={\mathbb F}_q^n$. However  the best results are obtained when all the codes in the chain  have been obtained by the same construction.
This is the case of some types of codes arising from algebraic geometry.

A similar bound can be stated for codes $\cod_I=\langle\{{\mathbf b}_i : i\in I\}\rangle$ where $I$ is an arbitrary subset of $\{1,\dots,n\}$ (without changing the order on the basis elements nor the map $\rho$). We leave this generalization as an exercise to the reader (or see \cite{f-rdecoding}).

\subsection{The Feng-Rao bound on the minimum distance of dual codes}\label{ctdul}

Given a linear $[n,k]$ code $\cod$, its {\em dual code} is defined as
\begin{equation*}
\cod^{\perp}=\{\bv\in\fq^n : \bc\cdot\bv =0 \mbox{ for all $\bc\in\cod$}\}
\end{equation*}
where $\cdot$ denotes the usual inner product in $\fq^n$
\begin{equation*}
\bu\cdot\bv=\sum_{i=1}^n u_iv_i.
\end{equation*} 
Then $\cod^{\perp}$ is a linear $[n,n-k]$ code. 
By using similar ideas to those explained in the previous subsection, we can give a bound on the minimum distance of dual codes. 
Let ${\mathcal B} =\{ {\mathbf b}_1, \dots, {\mathbf b}_n\}$ be a basis of ${\mathbb F}_q^n$ and consider the chain of dual codes
\begin{equation*} 
\cod_n^{\perp}=(\bcero)\subset\cod_{n-1}^{\perp}\subset\cdots\subset\cod_0^{\perp}=\fq^n.
\end{equation*}
Given a vector $\bu\in\fq^n$, define the {\em syndromes} of $\bu$
\begin{equation*}
s_1=s_1(\bu)=\bb_1\cdot \bu, \dots, s_n=s_n(\bu)=\bb_n\cdot \bu
\end{equation*} 
or equivalently $\bB\bu^T=\bs^T$, where $\bs=(s_1,\dots,s_n)$ and  $\bB$ is the matrix whose rows are the vectors $\bb_1,\dots,\bb_n$. Then $\bu\in\cod_r^{\perp}\setminus \cod_{r+1}^{\perp}$ if and only if $s_1=\dots=s_r=0$ and $s_{r+1}\neq 0$. Consider also the {\em two dimensional syndromes}
\begin{equation*}
s_{ij}=(\bb_i*\bb_j)\cdot\bu, \; 1\le i,j\le n.
\end{equation*}
Let $\bS$ be the matrix $\bS=(s_{ij})$, $1\le i,j\le n$. Note that this matrix can be written also as $\bS=\bB \bD(\bu) \bB^T$, where $\bD(\bu)$ is the diagonal matrix with $\bu$ in its diagonal. Since $\bB$ has full rank, we have $\rank (\bS)=\rank (\bD(\bu))=wt(\bu)$. 

\begin{lemma}\label{dual1}
Let $\bu\in\cod_r^{\perp}$. 
\begin{enumerate}
\item $s_{ij}=0$ for all $(i,j)$ such that $\rho_{\mathcal B}(\bb_i*\bb_j)\le r$.
\item If $\bu\not\in \cod_{r+1}^{\perp}$ then $s_{ij}\neq 0$ for all $(i,j)$ such that $\rho_{\mathcal B}(\bb_i*\bb_j)= r+1$.
\end{enumerate}
\end{lemma}
\begin{proof}
As $\bu\in\cod_r^{\perp}$ we have $s_1=\dots=s_r=0$. (1) If $\rho_{\mathcal B}(\bb_i*\bb_j)\le r$ then, according to Lemma \ref{properties}(2),  $\bb_i*\bb_j=\lambda_1\bb_1+\dots+\lambda_r\bb_r$ and $s_{ij}=\lambda_1 s_1+\dots+\lambda_r s_r=0$. (2) If $\bu\not\in \cod_{r+1}^{\perp}$ then $s_{r+1}\neq 0$. When $\rho_{\mathcal B}(\bb_i*\bb_j)= r+1$, we have $\bb_i*\bb_j=\lambda_1\bb_1+\dots+\lambda_r\bb_r+\lambda_{r+1}\bb_{r+1}$ with  $\lambda_{r+1}\neq 0$. Then 
$s_{ij}=\lambda_1 s_1+\dots+\lambda_r s_r+\lambda_{r+1}s_{r+1}=\lambda_{r+1}s_{r+1}\neq 0$.
\end{proof}

For $r=0,\dots,n-1$, define the sets
\begin{equation*}
N_r=\{ (i,j) : \mbox{$(\bb_i,\bb_j)$ is well-behaving and $\rho_{\mathcal B}(\bb_i*\bb_j)= r+1$} \}.
\end{equation*}

Let $N_r=\{(i_1,j_1),\dots,(i_t,j_t)\}$. The well-behaving property implies that all $i$'s in this set are distinct. Write $i_1<i_2<\dots<i_t$. By symmetry, $j_t=i_1,\dots,j_{1}=i_t$, hence $j_t<\dots<j_1$. Let $\bS_r$ be the submatrix of $\bS$
\begin{equation*}
\bS_r=
\left[ 
\begin{array}{ccc}
s_{i_1,j_t} & \cdots & s_{i_1,j_1} \\
\vdots & & \vdots \\
s_{i_t,j_t} & \cdots & s_{i_t,j_1} \\	
\end{array}
\right].
\end{equation*}

\begin{lemma}\label{dual2}
If $\bu\in \cod_{r}^{\perp}\setminus\cod_{r+1}^{\perp}$ then $\bS_r$ has full rank.
\end{lemma}
\begin{proof}
Let $(l,m)$ be an entry in the anti-diagonal of $\bS_r$. Then $l=i_h, m=j_h$ for some $h$ and $s_{lm}\neq 0$ by Lemma \ref{dual1}(2). If $(l,m)$ is above the anti-diagonal, then  $l=i_h, m<j_h$, hence $\rho_{\mathcal B}(\bb_l*\bb_m)< \rho_{\mathcal B}(\bb_{i_h}*\bb_{j_h})= r+1$. Thus $s_{lm}=0$ by Lemma \ref{dual1}(1) and $\det(\bS_r)\neq 0$.
\end{proof}

As a consequence of this lemma, if $\bu\in \cod_{r}^{\perp}\setminus\cod_{r+1}^{\perp}$ we have
$wt(\bu)=\rank(\bS)\ge \rank(\bS_r)=\# N_r$. The {\em Feng-Rao} or {\em  dual order bound} on the minimum distance of  $\cod_k^{\perp}$ {\em  with respect to the basis $\mathcal B$} states the following

\begin{theorem}\label{AndGeil-dual} 
For $k=0,1,\dots,n-1$, the minimum distance of $\cod_k^{\perp}$ satisfies 
\begin{equation*}
d(\cod_k^{\perp}) \geq  \min \{ \# N_r : r=k,\dots, n-1\}.
\end{equation*}  
\end{theorem}

As in case of primary codes, this bound depends on the choice of the basis $\mathcal B$.

\section{Evaluation codes and order domains}

The theory introduced in the previous section directly leads to the problem of finding basis $\mathcal B$ producing good codes. This subject will  be addressed in this section.

Let $R$ be a $\fq$-vector space and let $\Phi$ be a linear map $\Phi : R\rightarrow\fq^n$. For every linear subspace $L\subseteq R$ we have a linear code
\begin{equation*}
\cod(L)=\Phi(L)
\end{equation*}
and its dual
\begin{equation*}
\cod(L)^{\perp}=\{\bv\in\fq^n : \bc\cdot\bv=0 \; \mbox{ for all $\bc\in\cod(L)$}\}.
\end{equation*}
If we consider a basis $\{f_1,f_2,\dots \}$ of $R$, we get a chain of linear codes, $\cod_r=\langle \Phi(f_1),\dots,\Phi(f_r)\rangle$, $r=1,2,\dots$. When $\Phi$ is surjective, then there exists $r$ such that $\cod_r = \fq^n$, and the order bounds can be applied to obtain  estimates on the minimum distance of these codes.

\subsection{Evaluation codes}

The most interesting case of the above construction arises when $R$ is a set of functions that can be evaluated at points $P_1,\dots,P_n$ belonging to  a  geometrical object $\calx$. Set $\mathcal{P}=\{P_1,\dots,P_n\}$ and let $\Phi=ev_{\calp}:R\rightarrow\fq^n$ defined by $ev_\calp(f)=(f(P_1),\dots,f(P_n))$. The obtained codes are called {\em evaluation codes}. 

\begin{example} [Reed-Muller codes]
To give a concrete example
take the  $\fq$-algebra $R = \fq[X_1, \dots ,X_m]$ and let $\calp$ be the set of all $n=q^m$ points $P_1,\dots ,P_n$ in $\fq^m$. The evaluation map
\begin{equation*}
ev_\calp: \fq[X_1, \dots ,X_m]\rightarrow \fq^n
\end{equation*}
$ev_\calp(f) = (f(P_1),\dots , f(P_n))$, is  linear and verifies  $(f g)(P) = f(P)*g(P)$, so it is a morphisms of $\fq$-algebras. To see surjectivity, given a point $P = (a_1, \dots, a_m)\in \fq^m$,  the polynomial
\begin{equation*}
f_P =\prod_{i=1}^{m}\prod_{\alpha\in\fq, \alpha\neq a_i} (X_i - \alpha)
\end{equation*}
verifies $f_P(P)\neq 0$ and $f_P(Q)=0$ for all $Q\neq P$. Thus the vectors $\{ ev_\calp(f_P): P\in\fq^m\}$ span $\fq^n$.
Consider the basis $\{ f_1,f_2,\dots\}$ of $\fq[X_1,\dots,X_m]$ consisting of all monomials ordered according to a graded order (for example the graded lexicographic order: first compare degrees;  then  apply lexicographic order to break ties). Then we obtain an increasing chain of codes $\cod_1\subset \cod_2\subset\dots$, where
\begin{equation*}
\cod_i=ev_\calp (\langle f_1,\dots,f_i\rangle).
\end{equation*}
Among these codes, particular interest have the ones of the form
$\mathcal{R}\mathcal{M}(r,m)=ev_\calp(\fq[X_1,\dots,X_m]_{(r)})$,
where $\fq[X_1,\dots,X_m]_{(r)}$ stands for the linear space of all polynomials of degree at most $r$. They are called {\em Reed-Muller} codes.  The same construction can be done by considering  homogeneous polynomials and evaluating them at points in the projective space. In this case we obtain the so-called Projective Reed-Muller codes.  
\end{example}

Reed Muller codes are important from both theoretical and practical reasons and much is known about them. For example, in 1972 a Reed-Muller code was used by Mariner 9 to transmit black and white photographs from Mars. 
The case $m=1$  is particularly simple and interesting, so it deserves a special attention.

\begin{example} [Reed-Solomon codes] \label{RSevaluation}
Let $R=\fq[X]$ and consider the basis $\{1,X,X^2,\dots\}$. Let $\calp$ be the set of points in the affine line $\fq$. The obtained evaluation codes, called {\em Reed Solomon} codes, are widely used  (CD players, bar codes, etc.). Their parameters are easy to obtain: as a polynomial of degree $r$ has at most $r$ roots, for $r<n$ the code $ev_\calp(\langle 1,X,\dots,X^r\rangle)$ has length $n=q$, dimension $k=r+1$ and minimum distance $d=n-r$ (it is a MDS code). 
\end{example}

In the above two examples, note that for all $f\in\fq[X_1, \dots ,X_m]$ it holds that $ev_\calp(f^q)=ev_\calp(f)$, hence we can obtain the same codes from the quotient algebra $\fq[X_1, \dots ,X_m]/\langle X_1^q-X_1,\dots,X_m^q-X_m\rangle$.
In general, we can take an ideal $I\subset\fq[X_1, \dots ,X_m]$ and consider $I_q=I+\langle X_1^q-X_1,\dots,X_m^q-X_m\rangle$. Let  $\calp=\{P_1, . . . , P_n\}$ be the set of all rational points in the zero set $V=V(I_q)$. The evaluation map $ev_\calp: R_q=\fq[X_1, \dots,X_m]/I_q\rightarrow \fq^n$ is a vector space isomorphism. For any linear subspace $L\subseteq R_q$ we define the {\em affine variety} code $C(I,L)=ev_\calp(L)$. It is known that every linear code can be obtained in this way. Also algebraic geometry codes from curves, which are the main subject of this chapter, are particular cases of this construction. Affine variety codes were introduced by Fitzgerald and Lax in \cite{FL}, where the reader can find more details.

\subsection{Weight functions and order domains}

In previous examples we have seen how to construct a chain of evaluation codes from an algebra $R$ and an ordered basis of $R$. The better this order, the better will be the results obtained when using the order bounds. We formalize this idea.

Let $\mathbb{N}_0={\mathbb N}\cup\{0\}$. A function $v: R\rightarrow \mathbb{N}_0\cup \{-\infty\}$ is a {\em weight} on $R$ if it verifies the following properties
 
\begin{enumerate}
\item[(W.1)] $v(f) = -\infty$ if and only if $f = 0$;
\item[(W.2)] $v(1) = 0$;
\item[(W.3)] $v(f + g) \le \max\{v(f), v(g)\}$;
\item[(W.4)] $v(fg) = v(f) + v(g)$;
\item[(W.5)] if $v(f) = v(g)$ then there exists an element $\lambda\in\fq^*$ such that $v(f  -\lambda g) < v(f)$.
\end{enumerate}

\begin{remark}\label{remark}
Let $v$ be a weight function on $R$. The following are simple consequences of properties (W.1) to (W.5).\newline
(a) For all $\lambda\in\fq^*$ we have $v(\lambda\lambda^{-1})=v(\lambda)+v(\lambda^{-1})=v(1)=0$. Then $v(\lambda)=0$. Conversely, if $v(f)=0$ then there exists $\lambda\in\fq^*$ such that $v(f-\lambda)=-\infty$ and $f=\lambda\in\fq$. \newline
(b) If $v(f)>v(g)$ then $v(f)=v(-g+(f+g))\le\max\{ v(g), v(f+g)\}=v(f+g)\le v(f)$, hence $v(f + g)=v(f)$. \newline 
(c) $R$ is an integral domain. If $fg=0$ with $g\neq 0$ then $v(1)\le v(g)$. Thus $v(f)\le v(fg)=-\infty$ which implies $v(f)=-\infty$ and so $f=0$.
\end{remark}

A $\fq$-algebra $R$ with a weight function $v$ will be called an {\em order domain}.
Let $H(v)=\{ v(f) : f\in R^*\}=\{v_1,v_2,\dots\}$ be the increasing sequence of all integers appearing as the order of a nonzero element. For each $v_i\in H(v)$ let $f_i\in R$ be such that $v(f_i)=v_i$ and consider the ordered set  $\mathcal{F}=\{ f_1,f_2,\dots\}$.

\begin{proposition}\label{baseorden}
Let $R$ be an order domain with order function $v$ and let  $\mathcal{F}=\{ f_1,f_2,\dots\}$  as above. Then 
\begin{enumerate}
\item $\mathcal{F}$ is a basis of $R$ over $\fq$.
\item If $f=\sum_j \lambda_jf_j$, then $v(f)=\max\{v(f_j): \lambda_j\neq 0\}$.
\end{enumerate}
\end{proposition}
\begin{proof}
An iterated application of property (W.5) shows that $\mathcal{F}$ is a basis of $R$. (2) follows from Remark \ref{remark} (b).
\end{proof}

\subsection{Semigroups}

A {\em numerical semigroup} is a set $S\subseteq \nn_0$ such that 
 (i) $0\in S$ and 
 (ii) if $a,b\in S$  then $a+b\in S$. 
Our interest on semigroups comes from the following fact, which is a consequence of properties (W.2) and (W.4).

\begin{proposition}
If $R$ is an order domain and $v$ is a weight function on $R$, then $H(v)$ is a numerical semigroup.
\end{proposition}

The elements of $S$ will be called {\em pole numbers} or just {\em poles}, while the elements in $\nn_0\setminus S$ will be called {\em gaps}. We shall denote by $\gaps(S)$ the set of gaps of $S$. The number $g=\#\gaps(S)$ is the {\em genus} of $S$. If $S$ has finite genus then the smallest integer $c$ such that $a\in S$ for all $a\ge c$ is the {\em conductor} of $S$. From now on, all the semigroups we consider will be of finite genus.

\begin{lemma}\label{symmetricsem}
The conductor $c$ of a semigroup of genus $g$ verifies $c\le 2g$.
\end{lemma}
\begin{proof}
Since $c-1$ is  a gap, given a pair $(a,b)\in\nn_0^2$ with $a+b=c-1$, at least one of these two numbers is also a gap. There are $c$ such pairs and $g$ gaps so we obtain the inequality.
\end{proof}

When  $c=2g$ the semigroup is called {\em symmetric}. Note that for symmetric semigroups, given a pair $(a,b)\in\nn_0^2$ with $a+b=c-1$, exactly one of these two numbers is a gap and the other is a pole. Conversely, this condition ensures that $c=2g$.

From Lemma \ref{symmetricsem},  the interval $[0,2g-1]$ contains $g$ poles and $g$ gaps. If we write $S$ as an increasing enumeration of its elements $S=\{v_1=0<v_2<\dots \}$, then $2g=v_{g+1}$, hence $v_{g+i}=2g+i-1$ for all $i=1,2,\dots$. The first nonzero element of $S$, $v_2$, is the {\em multiplicity} of $S$. It will play an important role in forthcoming sections of this chapter.

A set of generators of $S$ is a set $A=\{a_1,\dots,a_r\}\subset S$ such that any $a\in S$ can be written as a linear combination $a=\lambda_1a_1+\dots+\lambda_ra_r$ with nonnegative integer coefficients. In this case we write $S=\langle a_1,\dots,a_r\rangle$. All semigroups admit a finite set of generators. For example, the Ap\'ery set
\begin{equation*}
A(S)=\{a\in S^* : a-v_2\not\in S^*\} .
\end{equation*}

\begin{example}[Semigroups generated by two elements] \label{semigrupos2} 
Let $a,b\in\nn$, $a<b$. Let $\delta=\gcd(a,b)$. If $\delta\neq 1$ then $S=\langle a,b\rangle\subset \delta \nn_0$ is not of finite genus. Assume $\delta=1$.
From B\'ezout theorem, every integer $m$ can be written as $m=\lambda a+\mu b$. Adding and subtracting $ab$ to both summands if necessary, we can obtain an unique representation of this type with $0\le\mu<a$. Then $m$ is a pole when $\lambda\ge 0$ and a gap when $\lambda<0$. In particular, the largest gap is $c-1=-a+(a-1)b$.  Let us show that the semigroup is symmetric. Suppose the largest gap is the sum of two gaps $-a+(a-1)b=(\lambda_1 a+\mu_1 b)+(\lambda_2 a+\mu_2 b)$ with $\lambda_1,\lambda_2<0$, $0\le\mu_1,\mu_2<a$. Then $(-\lambda_1-\lambda_2-1)a=(\mu_1+\mu_2-a+1)b$. Since $-\lambda_1-\lambda_2-1>0$ we have $a|\mu_1+\mu_2-a+1<a$, a contradiction. Then the semigroup is symmetric and hence $c=2g$. $S$ has genus $g=(a-1)(b-1)/2$.
\end{example}

As a consequence of this example, a semigroup $S$ has finite genus if and only if the greatest common divisor of its nonzero elements is 1. In this case there exist $a,b\in S$ such that $\gcd (a,b)=1$ and  $\langle a, b \rangle \subseteq S$.

The following fact will be used several times in what follows.

\begin{lemma} \label{Ssetminusa}
Let $S$ be a semigroup of finite genus. If $a\in S$ then 
\begin{equation*}
\#(S\setminus (a+S))=a.
\end{equation*}
\end{lemma}
\begin{proof}
Let $c$ be the conductor of $S$ and $m$ an integer. If $m\ge a+c$ then $m\in S$ and $m\in a+S$. Thus $S\setminus (a+S) =U\setminus V$, where $U=\{ m\in S : m<a+c\}$ and $V=\{ a+m :  m \in S, a+m<a+c \}\subseteq U$. Clearly $\# U= a+c-g$ and $\# V=\# \{ m\in S : m<c\}=c-g$, where $g$ is the genus of $S$. Then $\#(S\setminus (a+S))=\# U-\# V=a$.
\end{proof}

\subsection{Codes from weights}\label{weights}

Let $R$ be an order domain over $\fq$ and $v$  a weight function on $R$. Let $H=H(v)=\{v_1,v_2,\dots\}$ be the semigroup of $v$. 
If $\delta=\gcd \{ a : a\in H(v)^*\}=1$ then the weight $v$ is called {\em normal}. Otherwise we define the normalization of $v$ as the weight $v'=v/\delta$.  From now on, all weight functions will be normal. 

For each $v_i\in H$ let $f_i\in R$ be such that $v(f_i)=v_i$. The ordered set  $\mathcal{F}=\{ f_1,f_2,\dots\}$ is a basis of $R$ as a vector space over $\fq$. For $m=-1,0,1,\dots$, we consider the linear subspaces
$$
L(m)=\{  f \in  R : v(f) \le m \}.
$$
Clearly  $L(-1)=(\bcero), L(0)=\fq$ and $\{ f_i : v_i\le m\}$ is a basis of $L(m)$. Then $L(m-1)\subseteq L(m)$ with equality if $m$ is a gap of $H$. Since $v$ is normal,  $H$ has a finite number of gaps, $g$. So equality occurs precisely $g$ times. If $m$ is a pole, then $\dim(L(m))=\dim(L(m-1))+1$.

Let $\Phi:R\rightarrow\fq^n$ be a surjective morphism of $\fq$-algebras (for example, an evaluation map). Then we obtain a chain of linear codes
\begin{equation}\label{chain}
(\bcero)\subseteq C(\Phi,0) \subseteq C(\Phi,1) \subseteq\dots
\end{equation}
where $C(\Phi,m)=\Phi(L(m))$. Since $\Phi$ is surjective, the chain contains exactly $n+1$ distinct codes. We define the {\em dimension set} of this chain as 
$$
M=M(\Phi,v)=\{ m\in \nn_0 : C(\Phi,m-1)\neq C(\Phi,m) \}.
$$
It is clear that  $M$ consists of $n$ integers. Write $M=\{ m_1=0,m_2,\dots,m_n \}$.  The name \lq\lq dimension" set of $M$ is justified by the following  fact.

\begin{proposition}uer
$\dim(C(\Phi,m_k))=k$. If $m$ is a nonnegative integer then 
$\dim(C(\Phi,m))=\max \{ r : m_r\le m\}$.
\end{proposition}
\begin{proof}
The first statement is clear. For the second one, if $m_k=\max \{ r : m_r\le m\}$ then $C(\Phi,m)=C(\Phi,m_k)$.
\end{proof}

Let $m$ be an integer. If $m\not\in H$ then $L(m)=L(m-1)$ hence $m\not\in M$. If $m\in H$, take $f\in R$ such that $v(f)=m$. Then $L(m)=L(m-1)+\langle f\rangle$ so $C(\Phi,m)=C(\Phi,m-1)+\langle \Phi(f)\rangle$. Then $m\in M$ if and only if $\Phi(f)\not\in C(\Phi,m-1)$. 

The conditions of being $\Phi$ a morphism and $v$ a weight, allow us to give estimates on the parameters of $C(\Phi,m)$. The ideal $(f)$ generated by $f$ is a linear subspace of $R$, hence we can consider the quotient ring $R/(f)$ as a vector space over $\fq$.

\begin{lemma}\label{gop1}
Let $f\in R$ be a nonzero element. If $v$ is a weight function on $R$ then 
$\dim (R/(f))=v(f)$.
\end{lemma}
\begin{proof}
The weight $v$ maps the ideal $(f)$ into the set $v(f)+H$. Let $f_1,f_2,\ldots \in R$ be such that $v(f_i)=v_i$ and $f_i\in (f)$ when $v_i\in v(f)+H$. Then $\{ f_1,f_2,\dots\}$ is a basis of $R$  and $\{ f_i+(f) : v_i\not\in v(f)+H\}$ is a basis of $R/(f)$.
Thus $\dim (R/(f))=\#(H\setminus(v(f)+H))=v(f)$ by Lemma \ref{Ssetminusa}.
\end{proof}

\begin{lemma}\label{gop2}
If $m<n$ then $L(m)\cap \ker(\Phi)=(0)$.
\end{lemma}
\begin{proof}
Let $f\in\ker (\Phi), f\neq 0$. Then $(f)\subseteq\ker(\Phi)$ and we have a well defined, linear, surjective map $\Phi: R/(f)\rightarrow \fq^n$. Thus $\dim (R/(f))\ge n$ and Lemma \ref{gop1} implies $v(f)\ge n$, hence $f\not\in L(m)$.
\end{proof}

\begin{proposition}
Let $m<n$ be a nonnegative integer. 
\begin{enumerate}
\item $m\in M$ if and only if $m\in H$.
\item The code $C(\Phi,m)$ has dimension $k=\dim(L(m))=\max\{ i : v_i\le m\}$ and minimum distance $d\ge n-m$. If the semigroup $H$ has genus $g$ and $2g\le m<n$, then $k=m+1-g$.
\end{enumerate}
\end{proposition}
\begin{proof}
If $m<n$ then the map $\Phi:L(m)\rightarrow\fq^n$ is injective by Lemma \ref{gop2}.
Then $m\in M$ if and only if $L(m-1)\neq L(m)$ that is if and only if $m\in H$. So $k=\dim (L(m))=\max\{ i : v_i\le m\}$. Since $H$ has $g$ gaps, its conductor verifies $c\le 2g$, so when $m\ge 2g$ we have $m=v_{m+1-g}$ implying $k=m+1-g$. 
Let us prove the statement about the minimum distance $d$. Let $\bc=\Phi(f)$, $f\in L(m)$, be a codeword of $C(\Phi, m)$ with weight $d$. Let $I=\{1,\dots,n\}\setminus \sop (\bc)$ be the set of zero coordinates of $\bc$ and $\pi:\fq^n \rightarrow \fq^{n-d}$ be the projection on the coordinates of $I$. The map $\pi\circ\Phi:R\rightarrow\fq^{n-d}$ is a surjective morphism of algebras. Since $f\in L(m)\cap\ker(\pi\circ\Phi)$, Lemma \ref{gop2} implies $m\ge n-d$ or equivalently $d\ge n-m$. 
\end{proof}

The inequality $d(C(\Phi,m))\ge n-m$ is  the {\em Goppa bound} on the minimum distance of $C(\Phi,m)$.

\subsection{The order and dual order bounds}

Besides the Goppa bound, we can apply to $C(\Phi,m)$ and its dual $C(\Phi,m)^{\perp}$ the bounds of Theorems \ref{AndGeil-primary} and \ref{AndGeil-dual} respect to the sequence $\cod_0=(\bcero) \subset \cod_1 \subset \dots \subset \cod_n$, obtained from the chain of equation \ref{chain}  after deleting repeated codes. Since $\dim(\cod_k)=k$, the map $\rho_{\mathcal B}$ defined in Section \ref{cotaAndersen-Geil} can be written as  
\begin{equation*}
\rho({\mathbf v})=\min \{\dim (C(\Phi,m)) : {\mathbf v} \in C(\Phi,m)\}.
\end{equation*}

\begin{lemma} \label{nle:3.1}
Let $f \in R^*$. 
\begin{enumerate}
\item $\rho(\Phi(f))\leq \dim C(\Phi,v(f))$ with equality if $v(f)\in M$.
\item If $v(f)\not\in M$ then $v(fh)\not\in M$ for all $h\in R^*$.
\end{enumerate} 
\end{lemma}
\begin{proof}
(1) The first statement is clear since $f\in L (v(f))$ and hence $\Phi(f) \in C(\Phi,v(f))$. If $v(f)\in M$ then $\Phi(f) \in C(\Phi,v(f))\setminus C(\Phi,(v(f)-1))$ and  $\rho(\Phi(f))= \dim C(\Phi,v(f))$.
(2) If $v(f)\not\in M$ then $\Phi(f) \in C(\Phi,v(f)-1)$ hence there exists $\psi\in L(v(f)-1)$ such that $\Phi(f)=\Phi(\psi)$. If $v(fh)\in M$ then $\dim C(\Phi,v(fh))=\rho(\Phi(fh))=\rho(\Phi(\psi h))\le \dim C(\Phi,v(\psi h))$. Since $v(fh)>v(\psi h)$ we get the equality $C(\Phi,v(fh))= C(\Phi,v(\psi h))$,  contradicting our assumption $v(fh)\in M$.
\end{proof} 

The equality $\rho(\Phi(f))= \dim C(\Phi,v(f))$  is not true in general. Let $\bar{H}=H\setminus M$. Lemma \ref{nle:3.1}(2) implies $\bar{H}+H\subseteq \bar{H}$, or equivalently $M\subseteq H\setminus (\bar{H}+H)$.

\begin{corollary}\label{LemaLewittesGeil}
$M\subseteq H\setminus (qH^*+H)$.
\end{corollary}
\begin{proof}
Let $m\in H$, $m\neq 0$, and let $f\in R$ be such that $v(f)=m$. Then $v(f^q)=qv(f)>v(f)$. Since $\Phi$ is a morphism, we have $\Phi(f^q)=\Phi(f)*\cdots*\Phi(f) \mbox{(q times)}=\Phi(f)$. Thus $qm\not\in M$. This proves $qH^*\subseteq\bar{H}$, so $qH^*+H\subseteq\bar{H}+H$ and $M\subseteq H\setminus (\bar{H}+H)\subseteq H\setminus qH^*+H)$.
\end{proof}
  
For $i=1,\dots,n$, let $\phi_i\in R$ be such that $v(\phi_i)=m_i$. The set ${\mathcal B}=\{ \Phi(\phi_1),\dots$, $\Phi(\phi_n) \}$ is a basis of ${\mathbb F}_q^n$ and the sequence of codes $(\cod_k)$ is given by
\begin{equation*}
\cod_k=\langle \Phi(\phi_1),\dots,\Phi(\phi_k) \rangle=C(\Phi,m_k), \; k=1,\dots,n.
\end{equation*}

\begin{proposition}\label{nwbp}
If $v_r+v_s=m_t\in M$ then $v_r,v_s\in M$ and  $(\Phi(f_r),\Phi(f_s))$ is a well-behaving pair with  $\rho(\Phi(f_r)*\Phi(f_s))=t$.
\end{proposition}
\begin{proof}
If $v_r+v_s\in M$, Lemma \ref{nle:3.1}(2) implies $v_r,v_s\in M$. Write $v_r=m_i, v_s=m_j$, so $\phi_i=f_r$ and $\phi_j=f_s$. We have
\begin{equation*}
\rho(\Phi(\phi_i)*\Phi(\phi_j))=\rho(\Phi(\phi_i\phi_j))=\dim C(\Phi,v(\phi_i\phi_j))=\dim C(\Phi,m_i+m_j).
\end{equation*}
If $(a,b)\prec (i,j)$ then $v(\phi_a\phi_b)<v(\phi_i\phi_j)$ and hence $\rho(\Phi(\phi_a)*\Phi(\phi_b))=\rho(\Phi(\phi_a\phi_b))< \dim C(\Phi,m_i+m_j)=\rho(\Phi(\phi_i)*\Phi(\phi_j)) $.
\end{proof}

From Proposition \ref{nwbp} we can derive a new version of the order bounds on the minimum distance of $C(\Phi,m)$ and $C(\Phi,m)^{\perp}$ as follows. For $r=1,\dots,n$, $s=0,\dots,n-1$, consider the sets 
\begin{equation*}
\Lambda^*_r =\{ (r,j) :  m_r+m_j \in M  \} \; , \;
N^*_s=\{ (i,j) :  m_i+m_j=m_{s+1}  \}
\end{equation*}
Define
\begin{eqnarray*}
d_{ORD}(k)&=&\min \{ \# \Lambda^*_r : r=1,\dots,k  \} \\
d^{\perp}_{ORD}(k)&=&\min \{ \# N^*_s : s=k, \dots,n-1  \}.
\end{eqnarray*} 
By applying the bounds of Theorems \ref{AndGeil-primary} and \ref{AndGeil-dual} with respect to the basis 
$\{\Phi(\phi_1),\dots$, $\Phi(\phi_n)\}$, we get the following result.

\begin{theorem} \label{nasterisk}
For a non-negative integer $m$, we have 
\begin{eqnarray*}
d(C(\Phi,m)) &\geq& d_{ORD}(\dim (C(\Phi,m))) \\
d(C(\Phi,m)^{\perp}) &\geq& d_{ORD}^{\perp}(\dim (C(\Phi,m))).
\end{eqnarray*}
\end{theorem}

The inequalities stated in this theorem are the {\em order} (or {\em Feng-Rao}) bounds on the minimum distances of the primary code $C(\Phi,m)$ and its dual $C(\Phi,m)^{\perp}$, respectively. They do not depend on the basis $\mathcal{B}$ but only on the dimension set $M$.

\subsection{Bibliographical notes}

Order domains and evaluation codes were introduced and studied by  T. H\o holdt, J.H. van Lint and R. Pellikaan, \cite{HLP}. 
The purpose was to simplify the theory of algebraic geometry codes and to formulate the order bound on the minimum distance in this language. This bound was first suggested by G.L. Feng and T.N.T. Rao in  \cite{FengRao} for the duals of one-point algebraic geometry codes. 
At the same time, R. Matsumoto and S. Miura independently developed many of the
same ideas for duals of one-point codes. They also formulated the Feng-Rao bound for any linear code defined by means of its parity check matrix, \cite{Miura}.
Another generalization to all linear codes
described by means of generator matrices, was given by Andersen and Geil, \cite{AG}.
That paper is primarily devoted to linear codes, but also the cases of codes from
order domains and affine variety codes
are treated. This is the bound we have stated in Theorem \ref{AndGeil-primary}. 
Many works have been devoted to study the relations between these bounds and to generalize them, see \cite{f-rdecoding} and the references therein.

Our presentation of order domains follows closely  \cite{HLP}.  In our exposition we have limited ourselves to consider weights $v$ whose semigroup $H(v)$ is a sub-semigroup of $\mathbb{N}_0$. 
If more general semigroups are allowed (for example,  subsemigroups of ${\mathbb N}_0^r$ for some $r$),  then the family of obtained codes is very enlarged. See \cite{affine,f-rdecoding}.

\section{Codes from Algebraic Geometry}

Some of the most interesting examples of evaluation codes are obtained from algebraic curves. This section is devoted to developing a basic introduction to  algebraic geometry codes.

\subsection{Algebraic curves}

It is not our intention here to explain the theory of algebraic curves, which can be found in many excellent books (eg.  \cite{fulton,HLP, StichtenothLibro}). Therefore we assume a certain familiarity of the reader with algebraic geometry and  we simply recall the basic ingredients we need to cook our codes.

An {\em algebraic curve} $\calx$ over $\fq$ is an absolutely irreducible algebraic variety of dimension one over $\fq$.
The set of rational points of $\calx$ is denoted $\calx(\fq)$. Algebraic geometry codes will be obtained through evaluation of rational functions of $\calx$ at (some)  points in $\calx(\fq)$, so we always refer to curves with $\calx(\fq)\neq\emptyset$.
Let $\fq(\calx)$ be the field of rational functions of $\calx$. Among all curves having $\fq(\calx)$ as a function field, there is (up to isomorphism) one nonsingular projective curve. We shall use this one for our code construction. Thus, in what follows, the word {\em curve} means an algebraic, projective, absolutely irreducible, nonsingular curve (although we eventually use singular plane models of such a curve for our computations).  

Points on $\calx$ correspond to valuation rings in its function field. Given a function $f\neq 0$, the {\em order} of $f$ at a point $P$ of $\calx$ is the integer $v_P(f)$, where $v_P$ is the discrete valuation corresponding to the valuation ring of $P$. If $v_P(f)<0$ then $P$ is a {\em pole} and if  $v_P(f)>0$ then $P$ is a {\em zero} of $f$. The divisor of $f$ is $\divi(f)=\sum_{P\in\calx} v_P(f)P$.

Given a rational divisor $G$ of $\calx$, we consider the vector space of functions having zeros and poles specified by $G$
\begin{equation*}
\call(G)=\{  f\in\fq(\calx) : \divi(f)+G\ge 0  \} \cup \{0\}.
\end{equation*}
The dimension of this space is denoted by $\ell(G)$.  Riemann-Roch theorem states that there is a constant $g$ (the {\em genus} of $\calx$) such that $\ell(G)=\deg(G)+1-g+\ell(W-G)$, where $W$ is a canonical divisor. Since  canonical divisors have degree $2g-2$,  it holds that $\ell(G)=\deg(G)+1-g$ when $\deg(G)>2g-2$.

Two divisors $G$ and $G'$ are {\em linearly equivalent}, denoted $G\sim G'$, if there is  rational function $\phi$ with $\mbox{div}(\phi)=G-G'$. In this case $\call(G)$ and $\call(G')$ are isomorphic via the map $f\mapsto \phi f$.

The gonality of the curve $\calx$ over $\fq$ is the smallest degree $\gamma$ of a non-constant morphism from $\calx$ to the projective line. Equivalently  $\gamma$ is the smallest degree of a
rational divisor $G$ such that $\ell(G) > 1$. More generally, the gonality sequence of $\calx$, $GS({\mathcal{X}})=\{\gamma_i : i=1,2,\dots\}$,  is defined by
\begin{equation*}
\gamma_i=\min\{ \deg(G) : \ell(G)\ge i \} .
\end{equation*}
Then $\gamma_1=0$ and $\gamma_2$ is the usual gonality. Since $\ell(G)\le \deg(G)+1$ when $\deg(G)\ge 0$, we have $\gamma_i\ge i-1$.  Conversely, from Riemann-Roch theorem it follows that  $\gamma_i\le i-1+g$ with equality for $i>g$.
The gonality sequence $GS({\mathcal{X}})$ verifies a symmetry property (similar to the symmetry property for semigroups):  for  every integer $r$, it holds that $r\in GS({\mathcal{X}})$ if and only if $2g-1-r\not\in GS({\mathcal{X}})$, cf. \cite{MTtrellis}. In general, computing $GS({\mathcal{X}})$ is a difficult task but for plane curves this sequence is entirely known and  depends only on the degree of $\calx$, see \cite{gonalityplane}.

\subsection{Algebraic geometry codes}

Let $\calx$ be a curve of genus $g$ over $\fq$ and let $\calp=\{P_1,\dots,P_n\}$ be a set of $n$ distinct rational points on $\calx$. Let $G$ be a rational divisor of nonnegative degree and support disjoint from $D=P_1+\dots+P_n$. The {\em algebraic geometry} code (or AG code) $C(\calx,D,G)$ is the image of the evaluation map
\begin{equation*}
ev_\calp :\call(G)\rightarrow\fq^n  \hspace*{1cm} ev_\calp(f)=(f(P_1),\dots,f(P_n)).
\end{equation*} 
$ev_\calp$ is a linear map whose kernel is $\call(G-D)$. The dimension of this kernel $a=\ell(G-D)$ is the {\em abundance} of $C(\calx,D,G)$.  In particular, if $\deg (G)<n$ then $a=0$ and hence $C(\calx,D,G)\cong \call(G)$. The parameters of this code are as follows.

\begin{theorem}\label{parametrosag}
The code $C(\calx,D,G)$ has dimension $k=\ell(G)-\ell(G-D)$ and minimum distance $d\ge n-\deg(G)+\gamma_{a+1}$. In particular, when $2g-2<\deg(G)<n$, then $k=\deg(G)+1-g$ and $d\ge n-\deg(G)$.
\end{theorem}
\begin{proof}
The statements about the dimension follow from the definition of $C(\calx,D,G)$ and the Riemann-Roch theorem. To see the bound on the minimum distance, let $\bc$ be a codeword of weight $d>0$.  Let $D'\le D$ be the divisor obtained as the sum of points in $\calp$ corresponding to the $n-d$ zero coordinates of $\bc$. There exist a function $f\in\call(G-D')\setminus \call(G-D)$ such that $\bc=ev_\calp (f)$. Then $\ell(G-D')\ge \ell(G-D)+1= a+1$ hence, by definition of gonality sequence, $\gamma_{a+1}\le \deg(G-D')= \deg(G)- (n-d)$.
\end{proof}

The weaker bound $d\ge d_G(C(\calx,D,G))=n-\deg (G)$ is often called the {\em Goppa bound} on the minimum distance.
Note that it is similar to the bound on the minimum distance of Reed-Solomon codes seen in Example \ref{RSevaluation} and the Goppa bound for codes coming from order domains.  The bound on $d$ stated in Theorem \ref{parametrosag}, $d\ge n-\deg(G)+\gamma_{a+1}$, is sometimes referred as the {\em improved Goppa bound}. 

\begin{proposition}\label{Goppaigual} 
$d(C({\mathcal{X}},D, G))=n-\deg(G)$  if an only if there exists a divisor $D', 0\le D'\le D$ such that $G\sim D'$. 
\end{proposition}
     
\begin{proof} 
As in the proof of Theorem \ref{parametrosag}, $d=n-\deg(G)$ if an only if there exists a divisor $D', 0\le D'\le D$ such  that     
$\ell(G-D')>0$. Since $G$ and $D'$ have the same degree, this happens if and only if $G\sim D'$.
\end{proof}

From Theorem \ref{parametrosag}, the parameters of $C(\calx,D,G)$ verify $k+d\ge \ell(G)-\deg(G)+n$. According to Riemann-Roch theorem,  a simple computation shows that this inequality implies
\begin{equation}\label{Singdefect}
n+1-g\le k+d \le n+1
\end{equation}
where the right-hand inequality is the  Singleton bound. The number $n+1-k-d$ is the {\em Singleton defect} of $C(\calx,D,G)$. Recall that $n+1-k-d\le g$ and that codes of Singleton defect 0 are MDS.

\begin{example}
Take $\calx=\pp^1$ the projective line over $\fq$. Let $Q$ be the point at infinity and $\calp$ the set of $n=q$ affine points. Then $C(\pp^1,D,mQ)$, $1\le m\le q$, is precisely the Reed-Solomon code of dimension $k=m+1$. Since $g=0$, it is a MDS code.  
\end{example}

Thus AG codes can be seen as generalizations of RS codes: instead of the projective line $\pp^1$, consider an arbitrary curve $\calx$ over $\fq$. Note that
Reed-Solomon codes have excellent parameters $k$ and $d$,  but too small length (consider the case $ q = 2 $).
According to the Hasse-Weil bound, cf. \cite{StichtenothLibro}, we have
\begin{equation*}
|\# \calx(\fq)- (q+1)| \le 2g \sqrt{q}
\end{equation*}
hence longer codes can be obtained by using curves of higher genus, although then the Singleton defect increases.  From equation \ref{Singdefect}, the  relative parameters verify
\begin{equation*}
\frac{k}{n}+\frac{d}{n}\ge 1-\frac{g}{n}
\end{equation*}
so one way to get better codes from curves of high genus is to take $n$ large with respect to $g$. This strategy requires curves with many points respect to its  genus. 

\begin{example}[Codes on the Klein Quartic]\label{quartic}
Let us consider the curve $\calx$ defined over $\mathbb{F}_8$ by the projective equation $X^3Y+Y^3Z+Z^3X=0$. $\calx$ is called the {\em Klein quartic}. It is a nonsingular plane curve, hence its genus is 3 by  Pl\"ucker's formula. A direct inspection shows that $\calx$ has 24 rational points, which is the maximum possible number allowed by the Serre's improvement on the Hasse-Weil bound,
\begin{equation*}
|\# \calx(\fq)- (q+1)| \le g \lfloor 2\sqrt{q}\rfloor.
\end{equation*}
Consider the points $Q_0=(1:0:0), Q_1=(0:1:0), Q_2=(0:0:1)\in\calx(\mathbb{F}_8)$ and the divisor $G=m(Q_0+Q_1+Q_2)$, for $m=2,\dots,6$. Let $\calp$ be the set of 21 rational points different from $Q_1,Q_2,Q_3$ and let $D$ be the sum of all these points. The algebraic geometry code $C(\calx,D,G)$ was first studied in \cite{klein}. According to Theorem \ref{parametrosag} it has dimension $k=3m-2$ and minimum distance $d\ge 21-3m$. Note that for other values of $m$ the parameters of the obtained codes are much more difficult to estimate (try it!). For $m=3,4$, no codes are known improving these parameters, see \cite{mint}. Take, for example, $m=4$. Then $\ell(4(Q_0+Q_1+Q_2))=10$. The following ten functions
\begin{equation*}
\frac{X^3}{T},\frac{X^2Y}{T},\frac{X^2Z}{T},\frac{XY^2}{T},\frac{XYZ}{T}, 
\frac{XZ^2}{T},\frac{Y^3}{T},\frac{Y^2Z}{T},\frac{YZ^2}{T},\frac{Z^3}{T},
\end{equation*}
where $T=XYZ$,
belong to $\call(4(Q_0+Q_1+Q_2))$ and are linearly independent, hence they form a basis of $\call(4(Q_0+Q_1+Q_2))$. A generator matrix of $C(\calx,D,4(Q_0+Q_1+Q_2))$ is obtained by evaluating these functions at all points of $\calp$. 
\end{example}

\subsection{Isometric codes}

An {\em isometry} of $\fq^n$ is a linear map $l:\fq^n\rightarrow\fq^n$ leaving the Hamming metric invariant, $d(\bu,\bv)=d(l(\bu),l(\bv))$. Thus an isometry is an isomorphism. Two codes $\cod,\cod'$ of length $n$ are {\em isometric} if there is an isometry $l$ such that $l(\cod)=\cod'$. Clearly isometric codes have equal parameters $n,k,d$ and similar properties.

Let $\bx=(x_1,\dots,x_n)$ be a $n$-tuple of nonzero elements of $\fq^n$ and $\sigma\in{\mathcal S}_n$, the symmetric group on $n$ elements. The maps $\bx:\bv\mapsto\bx*\bv$ and $\sigma:\bv\mapsto (v_{\sigma(1)},\dots,v_{\sigma(n)})$ are isometries. Conversely,  it can be proved (and it is left as an exercise to the reader) that any isometry $l$ can be written as $l=\bx\circ\sigma$, where $\bx\in(\fq^*)^n$ and $\sigma\in {\mathcal S}_n$.

\begin{proposition}\label{isometriccodes}
Let $\sigma\in {\mathcal S}_n$ and  $D_{\sigma}=P_{\sigma(1)}+\dots+P_{\sigma(n)}$.
Let $G,G'$ be two rational divisors such that $\sop(G)\cap\calp=\sop(G')\cap\calp=\emptyset$. If $G\sim G'$ then the codes $C(\calx,D,G)$ and $C(\calx,D_{\sigma},G')$ are isometric.
\end{proposition}
\begin{proof}
If $G\sim G'$ then there exists a rational function $\phi$ such that $G-G'=\mbox{div}(\phi)$ and $\call(G)=\{\phi f : f\in \call(G')\}$. Thus  $C(\calx,D,G)=ev_{\calp}(\phi)*C(\calx,D,G')=ev_{\calp}(\phi)*\sigma^{-1}(C(\calx,D_\sigma,G'))$.
\end{proof}

A converse of Proposition \ref{isometriccodes} is also true under some supplementary conditions on $n$, see \cite{MunPel}.

\subsection{Duality}

The dual of an algebraic geometry code is again an AG code.

\begin{theorem}\label{dualcode}
There exists a differential form $\omega$ with simple poles and residue 1 at every point $P_i\in\calp$. If $W$ is the divisor of $\omega$, then
\begin{equation*}
C(\calx,D,G)^{\perp}=C(\calx,D,D+W-G).
\end{equation*}
\end{theorem}
\begin{proof}
(Sketch) The existence of such form $\omega$ is guaranteed by the independence of valuations, see \cite{StichtenothLibro}, Chapter I. The map $\call(D+W-G)\rightarrow \Omega(G-D)$, $\phi\mapsto \phi\omega$ is a  well defined isomorphism of vector spaces. Furthermore 
\begin{equation*}
\phi(P_i)=\phi(P_i) \mbox{res}_{P_i}(\omega)=\mbox{res}_{P_i}(\phi \omega)
\end{equation*}
where $\mbox{res}_{P}(\eta)$ denotes the residue at $P$ of the differential form $\eta$. Let $\bu\in C(\calx,D,G)$, $\bv\in C(\calx,D,D+W-G)$ and write $\bu=ev_{\calp}(f), \bv=ev_{\calp}(\phi)$. Then
\begin{equation*}
\bu\cdot\bv=\sum_{i=1}^n f(P_i) \phi(P_i)=\sum_{i=1}^n f(P_i) \mbox{res}_{P_i}(\phi \omega)=
\sum_{i=1}^n \mbox{res}_{P_i}(f\phi \omega).
\end{equation*}
Since $\mbox{div}(f)\ge -G$ and $\mbox{div}(\phi \omega)\ge G-D$, we have $\mbox{div}(f\phi\omega)\ge -D$, so $f\phi\omega$ has no poles outside $\mbox{sop}(D)$. Then
\begin{equation*}
\sum_{i=1}^n \mbox{res}_{P_i}(f\phi \omega)=\sum_{P\in\calx} \mbox{res}_{P}(f\phi \omega)=0
\end{equation*}
where the right-hand equality follows from the Residue theorem (\cite{StichtenothLibro}, Corollary IV.3.3). Finally, since 
$\mbox{dim}(C(\calx,D,G)) +\mbox{dim}(C(\calx,D,D+W-G))=n$, we get the result.
\end{proof}

\subsection{One-point codes and Weierstrass semigroups}

If $G$ is a multiple of a single rational point $Q$ of $\calx$ and $\calp$ is the set of rational points on $\calx$ different from $Q$,  then the code $C(\calx,D,mQ)$ is called {\em one-point}. These codes are, in general, 
easier to study than the others.

The space $\call(mQ)$ is the set of rational functions with poles only at $Q$ of order at most $m$.
The set of rational functions with poles only at $Q$
\begin{equation*}
\call(\infty Q)=\bigcup_{m=0}^{\infty} \call (mQ)
\end{equation*}
is an $\fq$-algebra. The evaluation map $ev_\calp$ is thus a morphism of $\fq$-algebras.   As the dimension of $C(\calx,D,(n+2g-1)Q)$ is $k=l((n+2g-1)Q)-l((n+2g-1)Q-D)=n$, we have $C(\calx,D,(n+2g-1)Q)=\fq^n$ and $ev_\calp$ is surjective. On the other hand, from the properties of valuations it follows that $-v_Q$ is a weight function on $\call(\infty Q)$ and  this algebra becomes an order domain. So the theory developed in Section \ref{weights} can be applied. In particular, the chain of codes stated in equation \ref{chain}, becomes 
\begin{equation*}
(\bcero)\subseteq  C(\calx,D,0)\subseteq C(\calx,D,Q)\subseteq C(\calx,D,2Q)\subseteq\cdots\subseteq C(\calx,D,mQ)\subseteq\cdots
\end{equation*}
For simplicity we shall write $v$ instead $-v_Q$ and $ev$ instead $ev_\calp$ whenever the point $Q$ and the set $\calp$ are fixed. Also in order to simplify the exposition
\begin{center}
{\em from now on we shall assume $n\ge 2g$}
\end{center}
(otherwise we must distinguish several cases, which makes the exposition very cumbersome).
The  semigroup associated to the weight $v$,
\begin{equation*}
H(v)=\{ v(f) : f\in \call(\infty Q), f\neq 0 \}
\end{equation*}
is now denoted $H(Q)$ and called the {\em Weierstrass semigroup} of $Q$. As it happens for general  weight functions, $m\in H(Q)$ iff $l(mQ)\neq l((m-1)Q)$ (and thus $l(mQ)=l((m-1)Q)+1$). Then, when $m$ is a gap we have $C(\calx,D,mQ)=C(\calx,D,(m-1)Q)$.
From Riemann-Roch theorem it holds that $l(2gQ)=g+1$ hence $H(Q)$ has the same genus $g$ as the curve $\calx$. Since $l((2g-1)Q)=g$, then $H(Q)$ is symmetric when $l((2g-2)Q)=g$, that is when $(2g-2)Q$ is a canonical divisor.

\begin{example}[Hermitian curves] \label{hermitianex1} 
Consider the curve $\calh$ defined over the field $\fqd$ by the affine equation
\begin{equation*}
y^q+y=x^{q+1}.
\end{equation*}
$\calh$ is called the {\em Hermitian curve}. Codes arising from this curve are the most studied among all AG codes. $\calh$ is a nonsingular plane curve, hence its genus is $g=q(q-1)/2$. Let us compute its rational points. $\calh$ has exactly one point at infinity $Q=(0:1:0)$, which is the common pole of $x$ and $y$. The map $\beta\mapsto \beta^q+\beta$ is the trace map from $\fqd$ to $\fq$ and hence it is $\fq$-linear and surjective. Let $\alpha\in\fqd$. Since $\alpha^{q+1}\in\fq$, we deduce that the polynomial $T^q+T=\alpha^{q+1}$ has $q$ different roots $\beta$ in $\fqd$. Then the line $x=\alpha$ intersects $\calh$ at $q$ different affine points, which are rational over $\fqd$. In terms of divisors
\begin{equation*}
\divi(x-\alpha)=\sum_{\beta\in\fqd,  \beta^q+\beta=\alpha^{q+1}} P_{\alpha,\beta} -qQ
\end{equation*}
where $P_{\alpha,\beta}=(\alpha:\beta:1)$.
A similar reasoning proves that when $\beta^q+\beta\neq 0$, we have
\begin{equation*}
\divi(y-\beta)=\sum_{\alpha\in\fqd, \alpha^{q+1}=\beta^q+\beta} P_{\alpha,\beta} -(q+1)Q.
\end{equation*}
In particular, from the first equality and since we have $q^2$ choices for $\alpha$, we deduce that $\calh$ has $q^3$ rational affine points, that is $q^3+1$ rational points in total. Then $\calh$ has the maximum possible number of rational points according to its genus as it achieves the Hasse-Weil upper bound. It is a {\em maximal} curve.

Let us compute the Weierstrass semigroup $H(Q)$. Once the divisors $\divi(x-\alpha)$ and $\divi(y-\beta)$ are known, we deduce that $q$ and $q+1$ are pole numbers, hence $\langle q,q+1\rangle\subseteq H(Q)$. According to Example \ref{semigrupos2}, the semigroup  $\langle q,q+1\rangle$ has genus $g=q(q-1)/2=g(\calh)$. Then we get equality  $H(Q)=\langle q,q+1\rangle$. In particular this semigroup is symmetric.
\end{example}

\begin{example}[Hermitian codes]\label{ExHermitian} 
One-point codes over $\fqd$ coming from Hermitian curves are called {\em Hermitian codes}. Let $Q$ be the point at infinity and $\calp$ be  the set of all $n=q^3$ affine points on $\calh$. Hermitian codes are the AG codes
\begin{equation*}
C(\calh,D,mQ)=ev (\call(mQ))
\end{equation*} 
$m=0,1,2,\dots$.
To describe these codes explicitly we must determine the spaces of rational functions $\call(mQ)$ and $\call(\infty Q)$. The Weierstrass semigroup can be a useful tool to accomplish this task. Write $H(Q)=\{v_1=0,v_2\dots\}$ as an increasing enumeration of its elements. A basis of $\call(\infty Q)$ is a set of functions $\{ f_i : i\in\nn\}$ such that $v(f_i)=v_i$, see Proposition \ref{baseorden}. If $m\in H(Q)$ then $m$ can be written as a linear combination $m=\lambda q+\mu (q+1)$, where $\lambda$ and $\mu$ are nonnegative integers and $\mu<q$. Then $v(x^\lambda y^\mu)=m$. It follows that a basis of $\call(\infty Q)$ is
\begin{equation*}
\{ x^\lambda y^\mu : 0\le \lambda, 0\le\mu<q \}
\end{equation*}
and  a basis of $\call(mQ)$ is
\begin{equation*}
\{ x^\lambda y^\mu : 0\le \lambda, 0\le\mu<q, \lambda q+\mu (q+1)\le m \}.
\end{equation*}
The parameters of these codes can be estimated from the arithmetic of $\calh$. For example, let us show that for small values of $m\in H$, the minimum distance of $C(\calh,D,mQ)$ attains the Goppa bound. Let $\alpha\in \fqd^*$ and let $\alpha_1,\dots\alpha_{q+1}$ be the roots of $T^{q+1}=\alpha^{q+1}$. These roots belong to $\fqd$ and are pairwise distinct, so we can write $\fqd=\{\alpha_1,\dots,\alpha_{q+1},\alpha_{q+2},\dots,\alpha_{q^2}\}$. Let $\beta_1,\dots,\beta_q$  be the roots of $T^q+T=\alpha^{q+1}$. Then for $i>q+1$, the affine points $(\alpha_i,\beta_j)$ are not in $\calh(\fqd)$. Let $\lambda,\mu$ be two integers such that $0\le \lambda<q^2-q, 0\le\mu<q$ and let $m=\lambda q+\mu(q+1)$. Then $m\in H, m<n$ and the function
\begin{equation*}
f=\prod_{i=1}^{\lambda} (x-\alpha_{q+1+i})\prod_{j=1}^{\mu}(y-\beta_{j})
\end{equation*}
verifies $\divi(f)=D'-mQ$, with $0\le D'\le D$. Then, according to Proposition \ref{Goppaigual}, the code $C(\calh,D,mQ)$ attains the Goppa bound, $d(C(\calh,D,mQ))=n-m$. Since all poles $m\in H$ such that $m<n-q^2$ can be written in the form $m=\lambda q+\mu(q+1)$ with $0\le \lambda, 0\le\mu<q$, we deduce that all Hermitian codes $C(\calh,D,mQ)$ attain the Goppa bound for $m<n-q^2$. The same happens when $m<n$ is a multiple of $q$, $m=\lambda q$. To see that it is enough to consider the function
\begin{equation*}
f=\prod_{i=1}^{\lambda} (x-\alpha_{i}).
\end{equation*}
We shall compute the minimum distances of all nonabundant Hermitian codes later, seeing them as particular cases of Castle codes. 
\end{example}

The same reasoning as in the above example shows that for an arbitrary curve $\calx$ the ring $\call(\infty Q)$ is a finitely generated 
$\fq$-algebra. Take a generator set $\{a_1,\dots,a_r\}$ of $H(Q)$ and functions $\psi_1,\dots,\psi_r$ such that $v(\psi_i)=a_i$ for $i=1,\dots,r$. Then every element in $H(Q)$ is a combination of $a_1,\dots,a_r$ with nonnegative integer coefficients, hence $\call(\infty Q)=\fq[\psi_1,\dots,\psi_r]$.

\subsection{The dimension set and the order bound on the minimum distance}

Keeping the notation of previous sections, let ${\mathcal X}$ be a curve of genus $g$ defined over the finite field ${\mathbb F}_q$ and let $\calx(\fq)=\{Q,P_1,\dots,P_n\}$ be the rational points in ${\mathcal X}$. Let $\calp=\{P_1,\dots,P_n\}$.  Consider the chain of one-point codes  $(\bcero )\subseteq C(\calx,D,0)\subseteq\cdots\subseteq C(\calx,D,(n+2g-1)Q)=\fq^n$. 

The dimensions of these codes can be obtained from the dimension set $M=\{m_1,\dots,m_n\}$. 
Let $H=H(Q)= \{ v_1=0<v_2<\dots\}$ be the Weierstrass semigroup of $Q$ and  let $\gaps(H)=\{l_1,\dots,l_g\}$ be the set of gaps of $H$. Let us remember that
\begin{equation*}
M=\{ m\in{\mathbb N}_0 : C(\calx,D,mQ)\neq C(\calx,D,(m-1)Q) \}.
\end{equation*}

\begin{proposition}\label{ell}
$M= \{ m\in H : \ell(mQ-D)= \ell((m-1)Q-D)\}$.
\end{proposition} 
\begin{proof}
If $m\in M$ then $\ell(mQ)\neq \ell ((m-1)Q)$ and $m\in H$.  The kernel of the evaluation map restricted to $\mathcal L (mQ)$ is ${\mathcal L}(mQ-D)$, so 
when  $m<n$ this evaluation is injective and hence $m\in M$ if and only if $m\in H$.  When $m\geq n$  then $m-1,m\in H$ which implies $\ell(mQ)= \ell ((m-1)Q)+1$. Thus $C(\calx,D,mQ)\neq C(\calx,D,(m-1)Q)$ if and only if both kernels are equal.
\end{proof}

Thus, for all nonnegative integers $m<n$ we have  $m\in M$ if and only if $m\in H$. Then,  once $H$ is known, the problem of calculating $M$ is reduced to determine its last $ g $ elements.
Since $C(\calx,D,(n+2g-1)Q)=\fq^n$  we deduce that $g$ elements of  $\{ n,\dots, n+2g-1 \}$ belong to $M$ while the other $g$ elements do not.

 
  \begin{proposition}\label{sim}
If the divisors $D$ and $nQ$ are linearly equivalent, $D\sim nQ$, then  $M\cap\{ n,\dots,n+2g-1\}=\{ n+l_1,\dots,n+l_g \}$.
  \end{proposition}
  \begin{proof}
If $D\sim nQ$ then  $n\not\in M$ and  $n+v_1,\dots,n+v_g\not\in M$ by the remark after Lemma \ref{nle:3.1}. The statement follows by cardinality reasons.
  \end{proof}

\begin{example}[Hermitian codes]
As seen in Example \ref{DualExHermitian}, we have $D\sim nQ$. Then  Proposition \ref{sim} gives $M$.
\end{example}

We can obtain estimates on the minimum distance of one-point codes by using the order bound stated in Theorem \ref{nasterisk}:  
\begin{equation*}
 d(C(\calx,D,mQ))\geq d_{ORD}(\dim (C(\calx,D,mQ))).
\end{equation*} 
This bound improves the classical Goppa bound $d(C(\calx,D,mQ))\ge d_G(C(\calx,D$, $mQ))=n-m$  as the next result shows. Let $\pi$ be the smallest element in $\bar{H}=H\setminus M$.  Note that $\pi\ge n$. The sets $\Lambda^*_i$ can be rewritten as $\Lambda^*_i=\{m_j\in M : m_i+m_j\in M \}$ or, since $\bar{H}+H\subseteq \bar{H}$  as noted after Lemma \ref{nle:3.1}, as $\Lambda^*_i=\{m\in M : m-m_i\in H \}=(m_i+H)\cap M$.

\begin{proposition}\label{goppa}
For all $i=1,\dots,n$, we have $d_{ORD}(\dim (C(\calx,D,m_iQ)))\ge d_G(C(\calx,D,$ $m_iQ))$. If $m_i<\pi-l_g$ then  equality holds.
\end{proposition}
\begin{proof}
For the first statement it suffices to show  that $\#(M\setminus \Lambda_i^*)\le m_i$ for all $i$. Since $\Lambda^*_i=(m_i+H)\cap M$, we have $M\setminus \Lambda_i^*\subseteq H\setminus (m_i+H)$ and this follows from the fact that $\#(H\setminus (m_i+H))=m_i$, stated in Lemma \ref{Ssetminusa}.  If $m_i+l_g<\pi$, then all elements in $H\setminus (m_i+H)$ are smaller than $\pi$ and hence $M\setminus \Lambda_i^*= H\setminus (m_i+H)$. 
\end{proof}

\begin{example}[Codes on the Suzuki curve]\label{ExSuzuki}
The Suzuki curve $\cals$ is characterized as being the unique curve over $\fq$,
with $q = 2q_0^2$, and $q_0 = 2^r \ge 2$, of genus $g = q_0(q-1)$ having $q^2 + 1$ $\fq$-rational points, see \cite{Suz1}.
Without going into details, which would lead us too long,
a plane singular model of $\cals$ is given by the equation $y^q-y=x^{q_0}(x^q-1)$. Thus, there is just one point $Q$
over $x =\infty$  which is $\fq$-rational. The Weierstrass semigroup of $Q$ is known to be $H(Q) =\langle q, q + q_0, q + 2q_0, q + 2q_0 + 1 \rangle$  (see \cite{Suz2, Mathews}). 

Let us consider the particular case $q=8$. In this case the Suzuki curve  has genus $g=14$ and 65 rational points. A plane model of $\mathcal S$ is given by the equation  $y^8z^2-yz^{9} = x^2(x^8-xz^7)$.  This model is non-singular except at the point $(0 : 1 : 0)$. Being this singularity unibranched, the unique point $Q$ lying over $(0 : 1 : 0)$ is rational. Let us consider the codes $C(\cals,D,mQ)$, where $D$ is the sum of all 64 rational points of $\mathcal S$ except $Q$. The Weierstrass semigroup at $Q$ is 
\begin{eqnarray*}
H &=& \langle 8,10,12,13\rangle \\
   &=& \{ 0,8,10,12,13,16,18,20,21,22,23,24,25,26,28,\rightarrow\}.
\end{eqnarray*}
Then
\begin{eqnarray*}
qH^*+H &=& \{qv_i+v_j : v_i,v_j\in H, v_i\neq 0\}  \\
   &=& \{ 64,72,74,76,77,80, 82, 84,85,86,87,88,89,90,92,\rightarrow\}.
\end{eqnarray*}
By Corollary \ref{LemaLewittesGeil}, $M\subseteq H\setminus (qH^*+H)$, so we obtain
\begin{eqnarray*}
\lefteqn{M\subseteq\{ 0,8,10,\dots \mbox{(same as $H$)} \dots, 63, } \\
  & &  65,66,67,68,69,70,71,73,75,78,79,81,83,91 \}.
\end{eqnarray*}
Since both sets have cardinality $n=64$ we conclude that they are equal.
An straightforward computation gives the sequence $(\#\Lambda^*_i$, $1\le i\le 64)$: 
(64, 56, 54, 52, 51, 48, 46, 44, 43, 42, 41, 40, 39, 38, 36, 35, 34, 33, 32, 31, 30, 
29, 28, 28, 26, 25, 24, 23, 22, 21, 20, 21, 18, 19, 16, 17, 16, 13, 12, 14, 10, 13, 
8, 12, 10, 9, 8, 8, 6, 8, 7, 4, 5, 4, 4, 4, 5, 4, 3, 2, 2, 2, 2, 1). We find 14 non\-abundant codes $(m<64)$ for which the Goppa bound is improved (plus all the abundant ones). Specifically those corresponding to the values $m_i\in \{37, 45, 47, 49, 50, 53$,  $55, 57, 58, 59, 60, 61, 62, 63\}$. In particular we find four codes $[64,37,\ge 16],[64,58,\ge 4],[64,62,\ge 2]$ and $[64,63,\ge 2]$ achieving the best known parameters, see \cite{mint}.
\end{example}

\subsection{Duals of one-point codes}

The dual of an one-point code is not one-point in general. According to Proposition \ref{dualcode} we have $C(\calx,D,mQ)^{\perp}=C(\calx,D$, $D+W-mQ)$, where $W$ is the divisor of a differential form $\omega$ with simple poles and residue 1 at all points $P_i\in\calp$. Then we have the following result.

\begin{proposition}\label{dualll}
If there exist a differential form $\omega$ with simple poles and residue 1 at all points $P_i\in\calp$, such that $\mbox{\rm div}(\omega)=(n+2g-2)Q-D$ then 
$C(\calx,D,mQ)^{\perp}=C(\calx,D,(n+2g-2-m)Q)$.
\end{proposition}

In this case, the dual of an one-point code $C(\calx,D,mQ)$ is again an one-point code, $C(\calx,D,mQ)^{\perp}=C(\calx,D,(n+2g-2-m)Q)$. Thus we get two order bounds on the minimum distance of this code, namely
$d_{ORD}(\dim C(\calx,D,mQ))$ and $d^{\perp}_{ORD}(\dim C(\calx,D,(n+2g-2-m)Q))$. Both bounds give the same result.

\begin{proposition}
If there exist a differential form $\omega$ with simple poles and residue 1 at all points $P_i\in\calp$, such that $\mbox{\rm div}(\omega)=(n+2g-2)Q-D$, then
$d_{ORD}(\dim C(\calx,D,mQ))=d^{\perp}_{ORD}(\dim C(\calx,D,(n+2g-2-m)Q))$.
\end{proposition}

The proof of this result can be found in \cite{GMRT}.

\begin{example}[Duals of Hermitian and Suzuki codes]\label{DualExHermitian}
Consider the Hermitian curve $\calh$ over $\fqd$. 
The function
\begin{equation*}
f=\prod_{\alpha\in\fqd} (x-\alpha)
\end{equation*}
has divisor $\divi(f)=D-q^3Q$, where $D$ is the sum of all $n=q^3$ rational affine points on $\calh$. Then $\mbox{\rm div}(f)=D-nQ$. It can be proved (see \cite{StichHermitian}) that $\mbox{div}(df/f)=(n+2g-2)Q-D$. Thus $C(\calh,D,mQ)^{\perp}=C(\calh,D,(n+2g-2-m)Q)$. 
Analogously, for the Suzuki curve ${\mathcal S}$ over $\fq$, the function
\begin{equation*}
f=\prod_{\alpha\in\fq} (x-\alpha)
\end{equation*} 
verifies $\mbox{\rm div}(f)=D-nQ$ and  $\mbox{div}(df/f)=(n+2g-2)Q-D$.  Then the dual of an one-point Suzuki code is one-point too.
\end{example}

\subsection{Improved codes}\label{Sect_improved}

By choosing suitable functions $ f $ to be evaluated, in some cases we can slightly change one-point codes improving their parameters.
Let $\delta$ be an integer, $0<\delta\le n$. Let $\calx,\calp, Q$ as in the previous sections. Given functions $ \phi_1,\dots,\phi_n$ such that $\phi_i\in {\mathcal L}(\infty Q)$ and $v(\phi_i)=m_i$, we define the {\em improved code}
\begin{equation*}
C(D,Q,\delta)=\langle \{ ev(\phi_i) : \#\Lambda^*_i\ge \delta  \}  \rangle .
\end{equation*}
From Proposition \ref{agbound} it is clear that the mi\-nimum distance of $C(D,Q,\delta)$ is at least $\delta$. 
The sequence $(\Lambda^*_i)$ is said to be {\em monotone} for $\delta$ if for every $i,j$ such that $\#\Lambda^*_i\ge \delta$ and $\#\Lambda^*_j<\delta$ we have that $i<j$. If $(\Lambda^*_i)$ is  monotone for $\delta$ then $C(D,Q,\delta)$ is an usual one-point code, so improved codes only improve one-point codes for those $\delta$ for which the sequence is not monotone. In this case the code $C(D,Q,\delta)$ depends on the choice of $ \phi_1,\dots,\phi_n$. In fact, if $\#\Lambda^*_i=\delta$ and $\#\Lambda^*_j<\delta$ for some $j<i$, then $v(\phi_i+\phi_j)=v(\phi_i)$ but in general $ev(\phi_j)\not\in C(D,Q,\delta)$, hence $ev(\phi_i+\phi_j)\not\in C(D,Q,\delta)$. Thus we have a collection of improved codes with designed distance $\delta$, depending on the collection of sets $\{ \phi_1,\dots,\phi_n \}$.
 
\begin{example}[Improved Suzuki codes]
Let us consider the Suzuki curve $\mathcal S$ over ${\mathbb F}_{8}$ of Example \ref{ExSuzuki}. In that example we computed the sequence $(\#\Lambda^*_i)$. 
This sequence is monotone for $\delta=3,5,6,9,13,14,18,20$, $21$. For example the one-point code $C(\cals,D,70Q)$ has dimension 55 and distance at least 4 (that is $d_{ORD}(55)=4$), whereas $C(D,Q,4)$ has the same distance and dimension 57.
\end{example}

\subsection{Bibliographical notes}

Algebraic geometry  codes (also called geometric Goppa codes) were introduced by V.D. Goppa in in the seventies, \cite{goppa1,goppa2}, as a generalization  of another family of codes previously invented by himself, that of classical Goppa codes.
AG codes became famous when M. Tsfasman, S.G. Vladuts and T. Zink showed in the early eighties, that there exist infinite families of these codes exceeding the Gilbert-Varshamov bound, \cite{TVZ}. 
The enormous interest aroused by these codes  has encouraged the study of the theoretical tools supporting them, mainly algebraic geometry over finite fields.

Codes coming from many interesting curves have been studied in detail. For what it is referring to the two main examples discussed in this chapter,
Hermitian codes were first studied by Stichtenoth, \cite{StichHermitian}, and later by many authors. Their minimum distances  were computed in \cite{YK} and their complete weight hierarchies in \cite{BM}. Suzuki codes were introduced by
J. P. Hansen and H. Stichtenoth, \cite{Suz2}. The true minimum distances of codes on this curve are known in many cases, but  not always.

Besides one-point codes,  which are the ones mainly discussed in this chapter, codes over more than one point (two, three or more) have been also studied, \cite{A1,A3,A2}. The interested reader can find multiple-point codes on the Hermitian curve \cite{MathewsH}, the Suzuki curve \cite{Mathews},  or  the Norm-Trace curve  \cite{NormTrace2points}.

Many works have been devoted to the study of the order bound for AG codes. In its original formulation this bound applies to the duals of  one-point codes. A nice generalization to
arbitrary AG codes was given by P. Beelen \cite{beelen} and later improved by I. Duursma,  R. Kirov and S. Park in a sequence of articles \cite{Du1,Du2,Du3}.
The application of Andersen-Geil bound to one-point
codes treated in this chapter is due to O. Geil, C. Munuera, D. Ruano and F. Torres, \cite{GMRT}.

\section{Castle curves and Castle codes}

As seen above, curves with many points with respect to  its genus provide codes with good parameters. This observation has led in recent years to an intensive research in order to determine good bounds on the number of rational points of a curve and to find curves with many points. For our purposes in this chapter is relevant one of these bounds, due to Lewittes. This bound has the particularity of being proved by using one-point codes.  It  links the number of points on the curve to the Weierstrass semigroup of one of them. This fact makes the bound  particularly interesting for coding theory because the properties of this semigroup  strongly affect the parameters of the obtained codes.

\subsection{The Lewittes bound on the number of rational points of an algebraic curve}\label{SeccionLewittes}

Let $\calx$ be a curve over $\fq$ and write $\calx(\fq)=\{ Q,P_1,\dots,P_n \}$, $\calp=\{P_1,\dots,P_n\}$. Consider the one-point codes $C(\calx,D,mQ)$.  Let $H=\{v_1=0,v_2,\dots \}$ be the Weierstrass semigroup of $Q$ and $v_2$ its multiplicity. 
 
\begin{theorem}[Lewittes-Geil-Matsumoto bound]\label{LewittesGeil}
Let $\calx$ be a curve over $\fq$, $Q$ a rational point and $H$ be the Weierstrass semigroup of $Q$. Then
 \begin{equation*}
\#{\mathcal{X}}({\mathbb{F}}_q)\leq \# (H\setminus (qH^*+H))+1\le  qv_2+1 
  \end{equation*} 
where $v_2$ is the multiplicity of $H$. 
\end{theorem}
\begin{proof}
Let $\calx(\fq)=\{ Q,P_1,\dots,P_n \}$, $\calp=\{P_1,\dots,P_n\}$, and consider the one-point codes $C(\calx,D,mQ)$.
Then $\#\calx(\fq)=n=\# M$. By Corollary \ref{LemaLewittesGeil}, $M\subseteq H\setminus (qH^*+H)$. Taking cardinalities we obtain the first inequality. To see the second one, note that $qv_2+H\subseteq qH^*+H$ and according to Lemma \ref{Ssetminusa} we have $\#(H\setminus (qv_2+H))=qv_2$.
\end{proof}

The bound $\#{\mathcal{X}}({\mathbb{F}}_q)\leq \# (H\setminus (qH^*+H))+1$ was stated by Geil and Matsumoto, \cite{GeilMatsumoto}, improving the previous result $\# {\mathcal{X}} ({\mathbb{F}}_q)\leq qv_2+1$ obtained by Lewittes, \cite{Lewittes}.

\subsection{Castle curves}

Let ${\mathcal{X}}$ be a curve over ${\mathbb{F}}_q$. $\calx$ is called {\em Castle} is there exists a rational point $Q\in\calx(\fq)$ such that: 
\begin{enumerate}
\item the Weierstrass semigroup of $Q$, $H(Q)$ is symmetric; and
\item the number of rational points on $\calx$ reaches the Lewittes bound $\#{\mathcal{X}}({\mathbb{F}}_q)= qv_2(Q)+1$
\end{enumerate}
where $v_2(Q)$ is the multiplicity of $H(Q)$. 

\begin{example}
Some of the curves previously discussed in this chapter are Castle. \newline
(1) A rational curve is clearly a Castle curve.  \newline
(2) The Hermitian curve $\calh$ over $\fqd$ is a Castle curve. Let $Q$ be the point at infinity. The Weierstrass semigroup  $H=\langle q, q+1\rangle$ is symmetric of multiplicity $v_2=q$ and $\# \calx(\fqd)=q^3+1$. \newline
(3) The Suzuki curve $\cals$ is Castle. Let $Q$ be the point over  $x =\infty$. The Weierstrass semigroup of $Q$,  $H(Q) =\langle q, q + q_0, q + 2q_0, q + 2q_0 + 1 \rangle$ is telescopic (see \cite{HLP}), hence  symmetric of multiplicity $v_2=q$. Since $\cals$ has $q^2 + 1$ rational points, it is a Castle curve.
\end{example}

Many of the most interesting curves for  Coding Theory purposes are Castle. Let us see other examples.

\begin{example}
Let $\calx$ be a hyperelliptic curve  and $Q$ a hyperelliptic rational point. $\calx$ is Castle if and only if $Q$ is the only rational hyperelliptic point on $\calx$ and $\calx$ attains equality in the hyperelliptic bound $\# \{ \mbox{rational nonhyperelliptic points}\} + 2\# \{ \mbox{rational}$  $\mbox{hyperelliptic points}\}$ $\le 2q+2$.
\end{example}

 \begin{example}[The Norm-Trace curve]\label{normtrace} 
Let us consider the curve defined over  ${\mathbb F}_{q^r}$ by the affine equation
  \begin{equation*}
x^{(q^r-1)/(q-1)}=y^{q^{r-1}}+y^{q^{r-2}}+\ldots+y
  \end{equation*}
or equivalently by $N_{{\mathbb F}_{q^r}\mid{\mathbb{F}}_q}(x)=T_{{\mathbb F}_{q^r}\mid {\mathbb{F}}_q}(y)$, where the maps $N$ 
and $T$ are respectively the norm and trace from ${\mathbb F}_{q^r}$ to ${\mathbb{F}}_q$. This curve has $2^{2r-1}+1$ 
rational points and the Weierstrass semigroup at the unique pole $Q$ of $x$ is given by
  \begin{equation*}
H(Q)=\langle q^{r-1}, (q^r-1)/(q-1)\rangle\, .
  \end{equation*} 
Since every semigroup generated by two elements is symmetric, this is a Castle curve. Codes 
on these curves have been studied by Geil in \cite{NormTrace}, where the reader can find proofs and details. 
\end{example}
  
\begin{example}[Generalized Hermitian curves]\label{example2.3} 
 For $r\geq 2$ let us 
consider the curve ${\mathcal{X}}_r$ over ${\mathbb F}_{q^r}$ defined by the affine 
equation 
 \begin{equation*}
   y^{q^{r-1}}+\ldots+y^q+y=x^{1+q}+\ldots +x^{q^{r-2}+q^{r-1}}\, 
 \end{equation*}
or equivalently by $s_{r,1}(y,y^q,\ldots,y^{q^{r-1}})=s_{r,2}(x,x^q,\ldots,x^{q^{r-1}})$, where 
$s_{r,1}$ and $s_{r,2}$ are respectively the first and second symmetric polynomials in $r$ variables. 
Note that ${\mathcal{X}}_2$ is the Hermitian curve. These curves were introduced by Garcia and Stichtenoth in 
\cite{GS}. They have $q^{2r-1}+1$ rational points. Let $Q$ be the only pole of $x$. Then 
$
H(Q)=\langle q^{r-1}, q^{r-1}+q^{r-2}, q^r+1\rangle
$. 
This semigroup is telescopic and hence symmetric (see e.g. \cite{HLP}). Therefore, ${\mathcal{X}}_r$ is 
a Castle curve.  AG-codes based on these curves were studied in 
\cite{Bulygin}  (binary case) and \cite{Sepulveda} (general case). 
\end{example}

The next proposition states a fundamental property of Castle curves.

\begin{proposition}\label{curveproperties1} 
Let ${\mathcal{X}}$ be a Castle curve with 
respect to a point $Q\in{\mathcal{X}}({\mathbb{F}}_q)$. Write $\calx(\fq)=\{Q,P_1,\dots,P_n\}$ and let $D=P_1+\dots+P_n$.
\begin{enumerate}
\item Let $f\in\call(\infty Q)$ be such that $v(f)=v_2$. For every $a\in\fq$ we have $\mbox{\rm div}(f-a)=D_a-v_2Q$ with $0\le D_a\le D$.
\item $D\sim nQ$.  
\end{enumerate}
\end{proposition}

\begin{proof}
(1) The morphism $f:\calx\rightarrow {\mathbb P}^1$ has degree $v_2$ hence $\# f^{-1}(a)\le v_2$ for all $a\in \fq$. Since $\# \calx(\fq)=qv_2$ we conclude that $\# f^{-1}(a)= v_2$. Then there exist exactly $v_2$ points $P\in \calx(\fq)$ such that $f(P)=a$. (2) Consider the one-point codes $C(\calx,D,mQ)$ and the function $\phi=f^q-f$. $v(\phi)=qv_2=n$ and $\phi(P_i)=0$ for all $P_i$. Then $\phi\in\call(nQ-D)$ hence $D\sim nQ$.  
\end{proof}

\begin{corollary}\label{dualcastle}
Let ${\mathcal{X}}$ be a Castle curve of genus $g$ with 
respect to a point $Q\in{\mathcal{X}}({\mathbb{F}}_q)$. Let $\calx(\fq)=\{Q,P_1,\dots,P_n\}$ and $D=P_1+\dots+P_n$. Then 
$(n+2g-2)Q-D$ is a canonical divisor. 
\end{corollary}
\begin{proof}
$(n+2g-2)Q-D\sim (2g-2)Q$. Since $H$ is symmetric this is a canonical divisor. 
\end{proof}

\begin{remark}
Let $\phi$ be the function defined in the proof of Proposition \ref{curveproperties1}. It can be proved that the differential form $\omega=d\phi/\phi$ has simple poles and residue 1 at all points $P_i$. So $\omega$ is the differential form for which we asked in Proposition \ref{dualcode}.
\end{remark}

Let us remember that by
$\gamma_r$ we denote the $r$-th gonality of ${\mathcal{X}}$ over ${\mathbb{F}}_q$. 

\begin{proposition}\label{curveproperties2} Let ${\mathcal{X}}$ be a Castle curve with 
respect to a point $Q\in{\mathcal{X}}({\mathbb{F}}_q)$ with Weierstrass semigroup $H=\{ v_1=0,v_2,\dots\}$. If the multiplicity at $Q$ satisfies $v_2\leq q+1$, then 
\begin{enumerate}
\item $\gamma_i\le v_i$ for all $i=1,2,\dots$. 
\item $\gamma_2=v_2$. 
\item $\gamma_i=v_i$ for $i\geq g-\gamma_2+2$. 
\end{enumerate}
\end{proposition}
\begin{proof} 
(1) Follows from the definition of gonality.
(2) There is a non-constant morphism of  degree $\gamma_2$ from $\calx$ to the projective line. Then $qv_2+1=\#\calx(\fq)\le \gamma_2 (q+1)$, so 
 $(qv_2+1)/(q+1)=v_2-(v_2-1)/(q+1)\leq \gamma\leq v_2$. By our hypothesis $v_2\le q+1$, it holds that $(v_2-1)/(q+1)<1$ and we get the equality. 
(3) The statement about the gonalities of high order follows from the fact that both, the semigroup 
$H$ and the set of gonalities $GS({\mathcal{X}})=(\gamma_r)_{r\geq 1}$ verify the same symmetry property: for 
every integer $t$, it holds that $t\in H$ (resp. $t\in GS({\mathcal{X}})$) if and only if $2g-1-t\not\in H$ 
(resp. $2g-1-t\not\in GS({\mathcal{X}})$).
\end{proof}

\subsection{Codes on Castle curves}\label{s3}

Let ${\mathcal{X}}$ be a Castle curve of genus $g$ over ${\mathbb{F}}_q$ with $(n+1)$ 
${\mathbb{F}}_q$-rational points, ${\mathcal{X}}({\mathbb{F}}_q)=\{Q,P_1\ldots,P_n\}$. A {\em Castle code} is a one-point code  $C({\mathcal{X}},D, mQ)$ constructed from $\calx$ and  $\calp=\{P_1,\dots,P_n\}$. 
Let $H=H(Q)=\{0= v_1<v_2<\ldots \}$  be the Weierstrass semigroup of $Q$. The dimension set $M$ can be easily obtained:  by Propositions \ref{sim} and \ref{curveproperties1}, $M=\{ m\in H : m<n\}\cup\{n+l_1,\dots,n+l_g\}= H\setminus (n+H)$.  
Define the function $\iota=\iota_Q: \nn_0\rightarrow\nn$ by $\iota(m)=\max\{ i :v_i\le m \}$.
Note that $\iota(m)=\ell (mQ)$.

\begin{proposition}
Let $m$ be a nonnegative integer. The Castle code  $C({\mathcal{X}},D$,  $mQ)$ has
dimension $k=\iota(m)-\iota(m-n)$ and abundance $\iota(m-n)$.
\end{proposition}

We now turn to the minimum distance. 

\begin{proposition}\label{dparan} 
Let $C({\mathcal{X}},D, mQ)$ be a Castle code. Then 
\begin{enumerate}
\item for $1\le m<n$, $C({\mathcal{X}},D, mQ)$ reaches Goppa bound if and only if $C({\mathcal{X}},D$, $(n-m)Q)$ does.
\item For $1\le r\le q-1$, $d(C({\mathcal{X}},D,rv_2Q)=n-rv_2$. 
\item For $n-v_2\le m\le n$, $d(C({\mathcal{X}},D,mQ)=v_2$.
\end{enumerate}
\end{proposition}
     
\begin{proof} 
(1) As seen in Proposition \ref{Goppaigual}, 
$C({\mathcal{X}},D, mQ)$ reaches equality in the Goppa bound if and only if then there exists $D', 0\le D'\le D$ such that $mQ\sim D'$. 
Let $D''=D-D'$. Thus $mQ\sim D-D''\sim nQ-D''$, hence $(n-m)Q\sim D''$ and the code $C({\mathcal{X}},D, (n-m)Q)$ also 
reaches equality in the Goppa bound.
(2) Follows from  Propositions \ref{Goppaigual} and \ref{curveproperties1}(1).
(3) $v_2=d(C({\mathcal{X}},D,(n-v_2)Q)\ge d(C({\mathcal{X}},D,mQ)\ge d(C({\mathcal{X}},D,nQ)\ge v_2$. The first equality comes from item (2) of this proposition and the last inequality is the improved Goppa bound on the minimum distance. 
\end{proof}

\begin{example}
The bound $d_{ORD}$ was computed for codes on the Suzuki curve over $\mathbb{F}_8$ in Example \ref{ExSuzuki}. In particular we found the result $d(C({\mathcal{S}},D, 62Q))\ge d(C({\mathcal{S}},D, 63Q))\ge 6$. By using Proposition \ref{dparan} we get now $d(C({\mathcal{S}},D, 62Q))= d(C({\mathcal{S}},D, 63Q))= 8$. So this last one is a $[64,50, 8]$ code and again we the get a code with the best known parameters according to \cite{mint}. Furthermore this fact  shows that the bound $d_{ORD}$ does not always improve on the improved Goppa bound $d(C({\mathcal{X}},D, mQ))\ge n-\deg(G)+\gamma_{a+1}$.
\end{example}

The cardinalities $\# \Lambda^*$ can be now computed in a simple way. 

\begin{lemma}\label{simetria}
For Castle codes it holds that $M=\{m\in H : n+2g-1-m\in H\}$.
As a consequence, $m_{n-r+1}=n+2g-1-m_r$ for $r=1,\dots,n$.
\end{lemma}
\begin{proof}
Let $m\in H$. From Riemann-Roch theorem, $\ell(mQ-D)=m-n+1-g+\ell((n+2g-2-m)Q)$, hence $\ell(mQ)=\ell((m-1)Q)$ if and only if $\ell((n+2g-2-m)Q)=\ell((n+2g-1-m)Q)$, that is if and only if $n+2g-1-m\in H$.
The conclusion $m_{n+1-r}=n+2g-1-m_r$ is clear.
\end{proof}

For  $i=1,\dots,n$, let $L_i=m_i+\gaps(H)=\{m_i+l_1,\dots,m_i+l_g\}$.

\begin{proposition}\label{final}
For Castle codes, $\#\Lambda^*_i=n-i+1-\# (L_i\cap M)$.
\end{proposition}
\begin{proof}
Since $M=\{ m\in H : m<n\}\cup\{n+l_1,\dots,n+l_g\}= H\setminus (n+H)$ and $H$  
is symmetric, we have $M=\{0,\dots,n+2g-1\}\setminus L$, where  $L=\{l_1,\dots,l_g, n+2g-l_g-1,\dots,n+2g-l_1-1\}$. For $i=1,\dots,n$, let
\begin{eqnarray*}
U_i&=&\{m_j\in M : m_i+m_j<n+2g, m_i+m_j \not\in M\}, \\
V_i&=&\{m_j\in M : m_i+m_j\ge n+2g\}. 
\end{eqnarray*}
Clearly $\#\Lambda^*_i=\#\{ m_j : m_i+m_j\in M\}=\#(M\setminus (U_i\cup V_i))=n-\# U_i-\# V_i$. Since $M\subset H$, we have $U_i=\{m_j\in M : m_i+m_j\in L\}=\{ n+2g-1-l_g-m_i,\dots,n+2g-1-l_1-m_i\}\cap M$. According to Lemma \ref{simetria}, $\# U_i=\# (L_i\cap M)$. Besides $\# V_i=i-1$. In fact, if $m_i+m_j\ge n+2g$, from Lemma \ref{simetria}, we can write $m_j=n+2g-1-m_t$ with $t=n-j+1$. Then $n+2g-1+m_i-m_t>n+2g-1$ if and only if $m_i>m_t$ and there exists $i-1$ such choices for $m_t$. 
\end{proof}  

Then for Castle codes we have 
\begin{equation*}
d(C(\calx,D,m_kQ))\ge d_{ORD}(k)=\min\{ n-r+1-\#(L_r\cap M) : r\le k\}.
\end{equation*}

\begin{example}[Hermitian codes]    
The minimum distances of Hermi\-tian codes $C(\calh,D,mQ)$ were computed in Example \ref{ExHermitian} for  $m$ in the range $0\le m\le n-q^2$. We shall study now the case $n-q^2<m<n$. Note that all $m$ in this range are pole numbers and  $n-m\le n-q^2$. Write $m=n-aq-b$ with $0\le a,b <q$. If $b\le a$ then $n-m\in H$ hence Proposition \ref{dparan}(1) and Example \ref{ExHermitian} ensure that $C(\calh,D,mQ)$ reaches the Goppa bound, $d(C(\calh,D,mQ))=d_G(C(\calh,D,mQ))=n-m=aq+b$. If $b>a$, then
\begin{eqnarray*}
d(C(\calh,D,(n-aq-a-1)Q))\le d(C(\calh,D,(n-aq-b)Q) \\ \le d(C(\calh,D,(n-(a+1)q)Q))=(a+1)q.
\end{eqnarray*}
A straightforward computation using Proposition \ref{final} shows that 
\begin{equation*}
d_{ORD}(C(\calh,D,(n-aq-a-1)Q))=(a+1)q
\end{equation*}
so we get equality,
$d(C(\calh,D,(n-aq-b)Q))=(a+1)q$.
\end{example} 

Finally we sate a duality property of Castle codes. As a consequence of Propositions \ref{isometriccodes}, \ref{dualll} and Corollary \ref{dualcastle}, we have the following.

\begin{proposition}\label{tralara}
For Castle codes, there exist $\bx\in(\fq^*)^n$ such that  $C(\calx,D$, $m_kQ)^{\perp}=\bx*C(\calx,D,(n+2g-2-m_k)Q)$ for all $k=1,\dots,n$.
\end{proposition}

Codes verifying the duality relation of the above proposition are called {\em isometry dual}.
Let ${\mathcal B} =\{ {\mathbf b}_1, \dots, {\mathbf b}_n\}$ be a basis of ${\mathbb F}_q^n$ such that  $C(\calx,D, m_rQ)=\langle {\mathbf b}_1, \dots, {\mathbf b}_r\rangle$, $r=1,\dots,n$. A vector $\bx\in(\fq^*)^n$ providing the isometries stated in the proposition  can be explicitly obtained from the duality relations, which lead to the system of linear equations $({\mathbf b}_i* {\mathbf b}_j)\cdot \bx=0$, $i+j\le n$. 

Since isometric codes have equal minimum distance, we can obtain estimates on the minimum distance of Castle codes by using both the order and dual order bounds. It can be proved that both bounds give the same result. 

\begin{proposition}\label{ult}
For Castle codes we have $\# N^*_{n-r}=\# \Lambda^*_r$, $r=1,\dots,n$. As a consequence 
\begin{equation*}
d_{ORD}(C(\calx,D,m_kQ))=\min\{ \# N^*_r : r=n-k,\dots,n-1\}.
\end{equation*}
\end{proposition}
\begin{proof}
According to Lemma \ref{simetria}, for Castle codes it holds that $m_{n+1-r}=n+2g-1-m_r$. Then
\begin{eqnarray*}
\# N^*_{n-r} &=& \#\{ (i,j) : m_i+m_j=m_{n-r+1}\} \\
          &=& \#\{ (i,j) : m_r+m_j=m_{n-i+1}\} \\
          &=& \#\{ (r,j) : m_r+m_j\in M\} \\
          &=& \# \Lambda^*_r.
\end{eqnarray*}
The conclusion is clear.
\end{proof}

\subsection{Bibliographical notes}

Castle curves and codes were introduced in \cite{castle1} and generalized in \cite{castle2}. The computation of  $d_{ORD}$ for some Castle codes (including all Hermitian and Suzuki codes) can be found in the article \cite{olaya}. For Hermitian codes this bound provides the true minimum distance of $C(\calh,D,mQ)$ for all $m$, see \cite{HLP}. Such distances were first computed by K. Yang and P.V. Kumar in \cite{YK} (without using order bounds).

\section{Feng-Rao decoding}

In this section we  show a very general decoding method for codes $\cod_k$ belonging to chains, as those treated in Section \ref{Sect1}. Keeping the notations used in that section, let ${\mathcal B}=\{{\mathbf b}_1,\dots,{\mathbf b}_n\}$ be a basis of $\fq^n$ and $\cod_r=\langle {\mathbf b}_1,\dots,{\mathbf b}_r\rangle$, $r=1,\dots,n$. By using the information given by the whole chain $\cod_0=(\bcero)\subset\cod_1\subset\cdots\subset\cod_n=\fq^n$ we can decode $\cod_k$.

If these codes are one-point AG codes, $\cod_r=C(\calx,D,m_rQ)$, then we take the basis vectors ${\mathbf b}_1=ev(\phi_1), \dots, {\mathbf b}_n=ev(\phi_n)$, where $v(\phi_r)=m_r$, as treated in previous sections.

\subsection{Preparation step}

Our decoding algorithm works for dual codes. hence we first
consider a dual basis $\cald=\{\bh_1,\dots,\bh_n\}$ of $\fq^n$ verifying  
\begin{equation*}
{\mathbf b}_i\cdot\bh_j=\left\{ \begin{array}{ll}
0         & \mbox{if $i+j<n+1$} \\
\neq 0    & \mbox{if $i+j=n+1$}
\end{array} \right.
\end{equation*}
where $\cdot$ stands for the usual inner product in $\fq^n$. These conditions imply the duality relations
\begin{eqnarray*}
\langle \bh_1,\dots,\bh_{n-r}\rangle=\cod_r^{\perp}=\langle {\mathbf b}_1,\dots,{\mathbf b}_r\rangle^{\perp}
\end{eqnarray*}
or equivalently
$\langle {\mathbf h}_1,\dots,{\mathbf h}_r\rangle^{\perp}=\cod_{n-r}$
for all $r=1,\dots,n$. If the chain $\cod_0=(\bcero)\subset\cod_1\subset\cdots\subset\cod_n=\fq^n$ verifies a duality relation $\cod_r^{\perp}=\cod_{n-r}$, $r=0,\dots,n$, then we take ${\mathbf h}_i={\mathbf b}_i$. If the chain verifies an isometry-dual relation $\cod_r^{\perp}=\bx*\cod_{n-r}$,  $r=0,\dots,n$ (the case of Castle codes), then we take
${\mathbf h}_i=\bx*{\mathbf b}_i$, $i=1,\dots,n$.

Once the basis $\cald$ has been fixed, we consider the dual chain 
\begin{eqnarray*}
\cod_n^{\perp}=(\bcero)\subset \cod_{n-1}^\perp \subset\cdots\subset \cod_{k+1}^{\perp}\subset\cod_k^{\perp}\subset\cdots\subset\cod_0^{\perp}=\fq^n
\end{eqnarray*}
and let $\rho_{\mathcal D}:{\mathbb F}_q^n \rightarrow \{0,\dots,n\}$ be the sorting map relative to the basis $\cald$,
defined by $\rho_{\mathcal D}({\mathbf v})=\min \{ i : {\mathbf v} \in \langle \bh_1,\dots,\bh_i\rangle\}$ if $\bv\neq \bcero$.  A pair of basis vectors $(\bh_r,\bh_s)$ is well-behaving with respect to $\cald$ if for all $(i,j)\prec (r,s)$ we have $\rho_{\cald}(\bh_i*\bh_j)<\rho_{\cald}(\bh_r*\bh_s)$. Remember that for $r=0,1,\dots,n-1$, we define the sets
\begin{equation*}
N_r=\{ (i,j) : \mbox{ $(\bh_i,\bh_j)$ is well-behaving with respect to $\cald$ and $\rho_{\cald}(\bh_i*\bh_j)=r+1$} \}.
\end{equation*} 
All these sets are precomputed in the preparation step. The dual order bound with respect to ${\mathcal D}$, stated in Theorem \ref{AndGeil-dual}, ensures that the minimum distance of $\cod_k=\langle {\mathbf h}_1,\dots,{\mathbf h}_{n-k} \rangle^{\perp}$ 
satisfies $d(\cod_k) \geq \delta= \min \{ \# N_r : r=n-k,\dots, n-1\}$.  
We can decode $\cod_k$  up to $(\delta-1)/2$ errors by using majority voting. 

When we consider one-point AG codes then we can manage the sets $N_r^*$ instead of $N_r$. If these codes are Castle, Proposition \ref{ult} implies that the Feng-Rao algorithm corrects errors of weight up to one half the order bound.

\subsection{Syndromes}

Let $\bu=\bc+\be$  be a received word, where $\bc\in \cod_k$ and $\be$ is  the error vector. Assume  $wt(\be)\le (\delta-1)/2$. To decode $\bu$ we shall compute the syndromes
\begin{equation*}
s_1=\bh_1\cdot \be, \dots, s_n=\bh_n\cdot \be.
\end{equation*} 
Consider the matrix $\bH$ whose rows are the vectors $\bh_1,\dots,\bh_n$. $\bH$ has full rank $n$ and $\bH\be^T=\bs^T$, where $\bs=(s_1,\dots,s_n)$. Once all  one-dimensional syndromes $s_i$ are known we can deduce the error vector by solving a system of linear equations. Note that $s_1,\dots,s_{n-k}$ can be derived from $\bu$: as $\cod_k^{\perp}=\langle {\mathbf h}_1,\dots,{\mathbf h}_{n-k}\rangle$, for $i=1,\dots,n-k$, we have
\begin{equation*}
\bh_i\cdot \bu =\bh_i\cdot(\bc+\be)=\bh_i\cdot\be=s_i.
\end{equation*}
In order to compute $s_{n-k+1},\dots,s_n$, we shall use two-dimensional syndromes
\begin{equation*}
s_{rt}=(\bh_r*\bh_t)\cdot\be, \; 1\le r,t\le n.
\end{equation*}
Let $\bS$ be the matrix $\bS=(s_{rt})$, $1\le r,t\le n$. As seen in Section \ref{ctdul},  this matrix can be written also as $\bS=\bH \bD(\be) \bH^T$, where $\bD(\be)$ is the diagonal matrix with $\be$ in its diagonal.  Since $\bH$ has full rank, we have $\rank (\bS)=\rank (\bD(\be))=wt(\be)$. For $1\le i,j\le n$ let us consider the submatrix of $\bS$
\begin{equation*}
\bS(i,j)=(s_{rt}), \; 1\le r\le i, 1\le t\le j.
\end{equation*}
An entry $(i,j)$ is a {\em discrepancy} of $\bS$ if $\rank(\bS(i-1,j-1))=\rank(\bS(i-1,j))=\rank(\bS(i,j-1))$ and
$\rank(\bS(i-1,j-1))\neq \rank(\bS(i,j))$. Clearly the total amount of discrepancies in $\bS$ is $\rank(\bS)=wt(\be)$.

\subsection{Computing unknown syndromes}

Assume that $s_1,\dots, s_l$ are known and $s_{l+1}$ is the smallest unknown syndrome. Let $(i,j)\in N_l$. The well-behaving property implies that for each $(r,t)\prec (i,j)$ we have $\rho_{\cald}(\bh_r*\bh_t)< \rho_{\cald}(\bh_i*\bh_j)=l+1$. Then there exist $\lambda_1,\dots,\lambda_l$ such that $\bh_r*\bh_t=\lambda_1\bh_1+\dots+\lambda_l\bh_l$ and $s_{rt}=\lambda_1s_1+\dots+\lambda_ls_l$. Thus the matrices $\bS(i-1,j-1), \bS(i-1,j)$ and $\bS(i-1,j-1)$ are known. If these three matrices have equal rank, then $(i,j)$ is called a {\em candidate}. Let $K$ be the number of discrepancies in the known part of $\bS$.  If $(r,t)$ is a known discrepancy, then all entries $(r,t')$ and $(r',t)$ with $r'>r, t'>t$ are noncandidates. Conversely, if $(i,j)\in N_l$ is not a candidate then there exists a known discrepancy in its same row or column.  Thus the number of pairs $(i,j)\in N_l$ which are not candidates is at most $2K$. 
If $wt(\be)\le (\#N_l-1)/2$, then
\begin{equation*}
\mbox{number of candidates} \ge \#N_l-2K \ge \# N_l-2wt(\be)>0
\end{equation*}
and there always exist candidates. Let $(i,j)$ be one of them. There is a unique value $s'_{ij}$ of entry $(i,j)$ such that $\rank(\bS(i-1,j-1))=\rank(\bS(i,j))$. The candidate $(i,j)$ is called {\em true} if $s'_{ij}=s_{ij}$  and {\em false} if $s'_{ij}\neq s_{ij}$. Since $s_{l+1}$ is unknown, then so is $s_{ij}$ and we cannot check in advance whether a candidate is true or false. However, a candidate $(i,j)$ is false if and only if it is a discrepancy, hence there are at most $wt(\be)$ false candidates in $\bS$. As $wt (\be)$ is 'small', most candidates will be true. Let us formalize this idea.

Let $T$ and $F$ be respectively the number of true and false candidates in $N_l$.  Since a false candidate is a discrepancy and  the total number of discrepancies is $wt(\be)$, we have $K+F\le wt(\be)\le (\# N_l-1)/2$. Combining this inequality with
\begin{equation*}
\# N_l=\# \mbox{candidates}+\# \mbox{noncandidates} \le (T+F)+2K
\end{equation*}
we obtain $F<T$ and the majority of candidates are true.

For each candidate $(i,j)$, compute $s'_{ij}$ and suppose $s_{ij}=s'_{ij}$. This assumption leads to a predicted value $s'_{l+1}$ of $s_{l+1}$ as above: since $\rho_{\mathcal D}(\bh_i*\bh_j)=l+1$, we can write $\bh_i*\bh_j=\lambda_1\bh_1+\cdots+\lambda_{l+1}\bh_{l+1}$ with $\lambda_{l+1}\neq 0$. Then $s_{ij}=\lambda_1s_1+\cdots+\lambda_{l+1} s_{l+1}$. Define the {\em vote} of  $(i,j)$ as $s'_{l+1}=\lambda_{l+1}^{-1}(s'_{ij}-\lambda_1s_1-\cdots-\lambda_{l} s_{l})$. 

Compute the votes of all candidates $(i,j)\in N_l$. Since the majority of candidates are true, we can derive the correct value of $s_{l+1}$ as the most voted among all candidates.

Once this value is known we proceed to the next unknown syndrome. If $wt(\be)\le (\delta-1)/2$ then $wt(\be)\le (\# N_l-1)/2$ for all $l=n-k,\dots,n-1$ and all syndromes $s_{n-k+1},\dots,s_{n}$ can be computed.  Assuming that all these sets $N_l$ have been precomputed, the complexity of this algorithm is that of solving a linear system of $n$ equations in $n$ unknowns, that is $O(n^3)$.

\subsection{Bibliographical notes}

The idea of using majority voting for unknown syndromes is due to G.L. Feng and T.N.T. Rao \cite{FR} and  I. Duursma, \cite{Duu1}.  The original algorithm was designed for duals of primary AG codes. A full and nice description for duals of codes coming from order domains can be found in \cite{HLP}. A generalization to a broad class of codes, including primary codes, was done in \cite{f-rdecoding}.  Our presentation is a mixture of these two works.

Decoding AG codes is a very active area of ​​research today. General AG codes $C(\calx,D,G)$ can be decoded by several methods.  Here  we just cite the nice report \cite{BH} by Beelen and H\o holdt, which is close to the ideas presented in this chapter.

\subsection{An example}

Let us consider the Hermitian curve $\calh:y^2+y=x^3$ defined over the field $\mathbb{F}_4=\{0,1,\alpha,\alpha^2\}$, where $1+\alpha=\alpha^2$. $\calh$ has genus 1 and nine rational points, namely $Q=(0:1:0)$ and the eight affine points
\begin{equation*}
\begin{array}{llll}
P_1=(0,0),& P_3=(1,\alpha),&   P_5=(\alpha,\alpha),& P_7=(\alpha^2,\alpha), \\ 
P_2=(0,1),& P_4=(1,\alpha^2),& P_6=(\alpha,\alpha^2),& P_8=(\alpha^2,\alpha^2).
\end{array}
\end{equation*}
Let $\calp=\{P_1,\dots,P_8\}$ and
consider the codes $C(\calh,D,mQ)$, $m=0,\dots,9$. The Weierstrass semigroup of $Q$ is $H=\langle 2,3\rangle=\{ 0,2,3,\rightarrow\}$, and the dimension set is $M=\{ 0,2,3,4,5,6,7,9\}$. Then,  a basis $\mathcal B$ of ${\mathbb F}_4^8$ is then given by the vectors
\begin{equation*}
\begin{array}{lccccccccccr}
{\mathbf b}_1 &=&ev_\calp(1)     &=&(1,&1,&1,     &1,       &1,       &1,       &1,       &1) \\
{\mathbf b}_2 &=&ev_\calp(x)     &=&(0,&0,&1,     &1,       &\alpha,  &\alpha,  &\alpha^2,&\alpha^2) \\
{\mathbf b}_3 &=&ev_\calp(y)     &=&(0,&1,&\alpha,&\alpha^2,&\alpha,  &\alpha^2,&\alpha,  &\alpha^2)\\
{\mathbf b}_4 &=&ev_\calp(x^2)   &=&(0,&0,&1,     &1,       &\alpha^2,&\alpha^2,&\alpha,  &\alpha) \\
{\mathbf b}_5 &=&ev_\calp(x y)   &=&(0,&0,&\alpha,&\alpha^2,&\alpha^2,&1,       &1,       &\alpha) \\
{\mathbf b}_6 &=&ev_\calp(x^3)   &=&(0,&0,&1,     &1,       &1,       &1,       &1,       &1) \\
{\mathbf b}_7 &=&ev_\calp(x^2 y) &=&(0,&0,&\alpha,&\alpha^2,&1,       &\alpha,  &\alpha^2,&1) \\
{\mathbf b}_8 &=&ev_\calp(x^3 y) &=&(0,&0,&\alpha,&\alpha^2,&\alpha,  &\alpha^2,&\alpha,  &\alpha^2)\\
\end{array}
\end{equation*}
In view of the duality property of Hermitian codes we can take ${\mathcal D}={\mathcal B}$. Consider the code $\cod=C(\calh,D,3Q)$ of dimension 3. A direct computation gives
\begin{eqnarray*}
\Lambda^*_1&=\{ (1, 1), (1, 2), (1, 3), (1, 4), (1, 5), (1, 6), (1, 7), (1, 8) \} \\
\Lambda^*_2&=\{ (2, 1), (2, 2), (2, 3), (2, 4), (2, 5), (2, 7) \} \\
\Lambda^*_3&=\{ (3, 1), (3, 2), (3, 3), (3, 4), (3, 6) \} \\
N^*_5&=\{      (1, 6), (2, 4), (3, 3), (4, 2), (6, 1) \} \\
N^*_6&=\{      (1, 7), (2, 5), (3, 4), (4, 3), (5, 2), (7, 1) \} \\
N^*_7&=\{      (1, 8), (2, 7), (3, 6), (4, 5), (5, 4), (6, 3), (7, 2), (8, 1) \} 
\end{eqnarray*}
hence both, the order and dual order bounds, ensure $d(\cod)\ge 5$, which is the true minimum distance of $\cod$ according to Example \ref{ExHermitian}. 
Then it can correct up to 2 errors.

Since $k=3$, the code $\cod$ allows us to encode $3$-tuples $\bz\in {\mathbb F}_4^3$ by $8$-tuples $\bc\in\cod$. Suppose we want to transmit the message $\bz=(1,1,1)$. It is encoded as $\bc=1 \bb_1+1\bb_2+1 \bb_3=(1,0,\alpha,\alpha^2,1,0,0,1)$. Suppose we receive the word 
$\bu=(0,0,\alpha,1,1,0,0,1)$ with error $\be=(1,0,0,\alpha,0,0,0,0)$.
To decode $\bc$
we first compute the known one-dimensional syndromes of $\be$
\begin{equation*}
s_1=\bb_1\cdot \be=\alpha^2, \;
s_2=\bb_2\cdot \be=\alpha , \;
s_3=\bb_3\cdot \be=1 , \;
s_4=\bb_4\cdot \be=\alpha , \;
s_5=\bb_5\cdot \be=1.
\end{equation*}
The smallest unknown syndrome is $s_6$. Using the information given by $s_1,\dots,s_5$ and $N^*_5$, the known part of $\bS$ is
\begin{equation*}
\bS=
\left[ 
\begin{array}{cccccccc}
\alpha^2 & \alpha & 1    & \alpha & 1        & * &    &   \\ 
\alpha   & \alpha & 1    & *      &          &   &    &  \\
1        & 1      & *    &        &          &   &    &  \\ 
\alpha   & *      &      &        &          &   &    &  \\
1        &        &      &        &          &   &    & \\
*        &        &      &        &          &   &    &  \\ 
         &        &      &        &          &   &    & \\
         &        &      &        &          &   &    & 
\end{array}
\right]
\end{equation*}
where the entries in $N^*_5$ are marked with $*$. Since $\rank (\bS(2,2))=2$ there is a unique candidate: $(3,3)$. As $s'_{3,3}=\alpha^2$ and $\bb_3*\bb_3=\bb_3+\bb_6$,  it votes for $s_6=s'_{3,3}-s_3=\alpha^2+1=\alpha$. 

Once this syndrome is known let us compute $s_7$. We first update the matrix
\begin{equation*}
\bS=
\left[ 
\begin{array}{cccccccc}
\alpha^2 & \alpha & 1        & \alpha & 1   & \alpha & * &  \\ 
\alpha   & \alpha & 1        & \alpha & *   &        &   &  \\
1        & 1      & \alpha^2 & *      &     &        &   &  \\ 
\alpha   & \alpha & *        &        &     &        &   &  \\
1        & *      &          &        &     &        &   &  \\
\alpha   &        &          &        &     &        &   &  \\ 
*        &        &          &        &     &        &   &  \\ 
         &        &          &        &     &        &   & 
\end{array}  
\right].
\end{equation*}
As above, the entries in $N^*_6$ are marked with $*$.
Candidates are $(3, 4)$ and  $(4, 3)$. A simple computation gives $s'_{3,4}=1, s'_{4,3}=1$, and both vote for $s_7=1$. 
Let us compute $s_8$. The  current form of $\bS$ is
\begin{equation*}
\bS=
\left[ 
\begin{array}{cccccccc}
\alpha^2 & \alpha & 1        & \alpha & 1        & \alpha & 1   & *  \\ 
\alpha   & \alpha & 1        & \alpha & 1        & \alpha & *   &    \\
1        & 1      & \alpha^2 & 1      & \alpha^2 & *      &     &    \\ 
\alpha   & \alpha & 1        & \alpha & *        &        &     &    \\
1        & 1      & \alpha^2 & *      &          &        &     &    \\
\alpha   & \alpha & *        &        &          &        &     &    \\ 
1        & *      &          &        &          &        &     &    \\ 
*        &        &          &        &          &        &     & 
\end{array}
\right].
\end{equation*}
Candidates are $ (3, 6), (4, 5), (5, 4)$ and $(6, 3)$. We get 
$ s'_{3, 6}=1, s'_{4, 5}=1, s'_{5, 4}=1, s'_{6, 3}=1$. All of them vote for $s_8=1$.

Once all one-dimensional syndromes are known, we deduce the error vector $\be$ by solving the system
$s_1=\bb_1\cdot \be, \dots, s_n=\bb_n\cdot \be$. In our case, as expected, $\be=(1,0,0,\alpha,0,0,0,0)$, hence 
$\bc=\bu-\be=(0,0,\alpha,1,1,0,0,1)-(1,0,0,\alpha,0,0,0,0)=(1,0,\alpha,\alpha^2,1,0,0,1)$. Finally we write $\bc$ as a linear combination of $\bb_1,\bb_2,\bb_3$, obtaining $\bc= \bb_1+\bb_2+ \bb_3$. The original message was $\bz=(1,1,1)$.

\bibliographystyle{amsplain}

\end{document}